%% file: arxiv.tex
\documentclass[runningheads]{llncs}

\usepackage[T1]{fontenc}
\usepackage{graphicx}
\usepackage{cite}

\makeatletter

\let\c@theorem\relax    
        
\let\c@lemma\relax      
    
\let\c@corollary\relax  
       
\let\c@remark\relax     
\makeatother
\usepackage{amsmath}
\usepackage{amssymb}
\usepackage{amsthm}
\usepackage{thmtools}
\usepackage{thm-restate}
\makeatletter
\let\thmt@theoremdefiner\thmt@original@newtheorem
\def\@thmcountersep{.}
\makeatother

\input{mycmd.tex}

\begin{document}

\title{\gen: Causal Compression for Sharded Blockchains}
\titlerunning{\gen: Causal Compression for Sharded Blockchains}
\author{Shir Cohen \and
Adam Alon \and
Raz Omessi \and
Amir Sarid \and
Dana Shamir \and
Ofir Zohar}

\authorrunning{S. Cohen et al.}
\institute{Gen Labs, Tel Aviv, Israel\\
\email{shir.co@gmail.com, adamalonil@gmail.com, romessi@gmail.com,\\
amir@sarid.io, dana.shamir01@gmail.com, ofir.z@genlabs.co}}
\maketitle              %
\begin{abstract}

We present \gen{}, a sharded blockchain that scales execution horizontally while allowing dependent cross-shard operations to progress within a single block.
Existing sharded systems typically place coordination between dependent cross-shard steps, making latency grow with causal depth.

\gen{} localizes atomicity to individual accounts and executes cross-account work asynchronously.
Using an \emph{execute-before-agree} architecture, dependent operations execute across shards as they become available, while consensus later validates and commits the resulting schedule.
This enables \emph{Causal Compression}: cross-shard latency need not grow with the causal depth of a computation.
\gen{} scales single-validator execution to 1.07M requests/s and
four-validator execution to 193K committed requests/s, while reducing
16-hop causal-chain latency by 8.2$\times$ versus deferred execution.

\end{abstract}

\input{intro_tight.tex}
\input{model.tex}

\input{cluster_new.tex}

\input{smr.tex}

\input{perf_tight.tex}

\input{eval_new.tex}

\input{related_tight.tex}

\begin{credits}
\subsubsection{\ackname}
We thank Iddan Kfir, Gal Sadeh, and David Lehavi for helpful discussions and valuable feedback on this work.
\end{credits}

\newpage
\bibliographystyle{splncs04}
\bibliography{ref}

\appendix
\input{appendix/appendix_algs.tex}

\input{appendix/appendix.tex}
\input{appendix/appendix_smr.tex}

\input{appendix/appendix_concise_discussion.tex}

\end{document}

%% file: mycmd.tex
\usepackage{amsthm}
\usepackage{graphicx}

\usepackage{thmtools, thm-restate}

\usepackage{tikz}

\newcommand{\circled}[1]{\tikz[baseline=(char.base)]{
    \node[shape=circle,fill=black,inner sep=1pt,text=white] (char) {\sffamily\textbf{\footnotesize #1}};}}

\usepackage{algorithm} 
\usepackage{algpseudocode} 

\usepackage{xcolor}

\usepackage{hyperref}

\usepackage{subcaption}

\newcommand{\gen}{GridSMR}
\newcommand{\algapp}{\ (Appendix~\ref{appendix:algorithms})}

\newtheorem{theorem}{Theorem}[section]
\newtheorem{corollary}{Corollary}[theorem]
\newtheorem{lemma}[theorem]{Lemma}
\newtheorem{assumption}[theorem]{Assumption}

\newtheorem{remark}[theorem]{Remark}

\algdef{SE}[UPON]{Upon}{EndUpon}[1]{\algorithmicupon\ #1\ \algorithmicdo}{\algorithmicend\ \algorithmicupon}
\algtext*{EndUpon}
\newcommand{\algorithmicupon}{\textbf{upon}}

\algdef{SE}[FOR]{ForAll}{EndForAll}[1]{\textbf{for all}\ #1\ \textbf{in parallel do}}{\textbf{end for}}
\algtext*{EndForAll}

%% file: intro_tight.tex
\section{Introduction}
\label{sec:intro}

The blockchain ecosystem continues its search for scalable Layer 1 (L1) infrastructure capable of supporting the next generation of decentralized applications, including latency-sensitive DeFi applications~\cite{gould2013limit, hyperliquid, dydx-v4}
and agentic payments~\cite{davidovic2026agentic}.
While consensus protocols have approached theoretical limits, execution remains a central scalability challenge.

We present \gen{}, a sharded smart-contract platform designed to scale execution horizontally without imposing a consensus delay on every dependent cross-shard step.
Independent accounts execute concurrently across shards, while chains of dependent cross-shard operations can progress within a single agreement interval.
As a result, \gen{} combines increasing execution parallelism with low-latency cross-shard computation.

Blockchains conventionally structure execution around transactions.
But ``transaction'' is not merely terminology: it imports ACID semantics from database systems~\cite{haerder1983principles}.
In particular, atomicity requires all state touched by a transaction to be updated all-or-nothing.
Sharding complicates this model by partitioning state across independently executing shards.
When a transaction spans multiple shards, preserving atomicity therefore requires coordination across them.

We challenge the assumption that ACID guarantees must be provided system-wide by the L1.
Large-scale distributed systems commonly narrow transactional scope through partitioned state ownership,
asynchronous communication, and weaker consistency models~\cite{pritchett2008base,thones2015microservices,vogels2009eventually}.
Blockchains have moved partially in this direction.
Guerraoui et al.~\cite{guerraoui2019consensus} explore such a model for cryptocurrency systems,
while Sui~\cite{sui, danezis2022narwhal} extends independent state ownership to a general-purpose L1.
However, Sui distinguishes owned from shared objects, with operations on shared objects still requiring coordinated ordering through consensus.

\gen{} adopts independently owned, asynchronous execution as the uniform execution model of the L1.
State is partitioned into \emph{accounts}, each of which owns its state and executes \emph{activations} atomically over that state.
Accounts do not directly modify one another.
Instead, computation across accounts proceeds by message passing: an activation may generate \emph{continuations} targeting other accounts.
Each continuation is an activation to be executed by its target account and may itself generate further continuations.
A client request can therefore unfold as a tree of causally dependent activations spanning many accounts and shards.
Atomicity is local to an activation rather than to the entire tree.

An activation tree is not a transaction:
if one continuation fails, preceding activations are not automatically unwound,
and applications requiring atomicity across accounts must implement that coordination themselves.
This weakens the default transactional semantics, but does not remove the ability to coordinate across accounts.
Applications can introduce stronger coordination where needed,
for example by serializing sensitive operations or placing state that must be updated atomically in the same account,
while leaving independent execution paths asynchronous.
Recent sharded AMM designs similarly partition state to gain parallelism while recovering
application-specific guarantees where necessary~\cite{chen2024samm,aanes2023automated}.
Such coordination reduces parallelism only for the affected operations,
rather than imposing that cost on every cross-account computation.
\gen{} therefore replaces mandatory system-wide atomicity with application-selective coordination.

This model shifts a different responsibility to the L1:
it must reliably carry asynchronous computation to completion.
In particular, if activation $A$ executes and generates continuation $B$,
then $B$ must eventually execute -- a property we call \emph{Causal Liveness}.
Standard SMR liveness does not imply it:
a system may continue processing new client requests while indefinitely failing to execute previously generated work.

Existing sharded blockchains also use asynchronous messages for cross-shard communication~\cite{ton,near}.
The key difference is when causally generated work becomes eligible for execution.
Consider three accounts $A$, $B$, and $C$ on different shards,
where $A$ invokes $B$, and the result of $B$ causes an invocation of $C$:
$
    A_{\text{shard 1}}
    \rightarrow
    B_{\text{shard 2}}
    \rightarrow
    C_{\text{shard 3}}.
    $
In receipt-based designs, $B$ can execute only after the step containing $A$ crosses an agreement boundary,
and $C$ must wait for another.
Each sequential cross-shard step therefore consumes an agreement interval,
making latency linear in causal depth.
\gen{} instead lets generated work execute as soon as it is delivered:
$A \rightarrow B \rightarrow C$ may all progress before the next agreement step.

\gen{} partitions accounts across execution shards and replicates the full set of shards across independently operated clusters.
Thus, sharding scales execution within each replica, while Byzantine replication operates across complete clusters.
The shards of a cluster are co-located,
so a continuation that crosses shards travels over a fast intra-cluster path,
while the slower inter-cluster path carries only replication and consensus traffic.
During each block interval, one cluster acts as leader and executes activations across its shards as they become available.
It produces per-shard traces recording client-submitted activations,
internally generated continuations, and their execution order.
Follower clusters independently re-execute and validate these traces before consensus finalizes the block.

\begin{figure}[t!]
    \centering
    \includegraphics[width=0.67\textwidth]{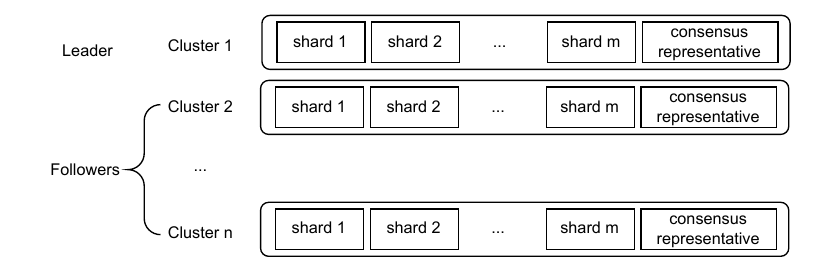}
    \caption{\gen{} architecture: clusters replicate execution across shards, with one leader cluster and multiple follower clusters.
    }
    \label{fig:gensmr}

\end{figure}

Recording the execution order is necessary because activation execution is deterministic given an account state and input,
while the concurrent cross-shard schedule is not:
continuations generated on different shards may arrive at a destination shard in different orders and thereby induce different states.
\gen{} therefore uses a \emph{speculative execute-before-agree} pipeline:
execution determines a candidate schedule,
and consensus commits that schedule only after independent validation.
The execution is provisional rather than final;
if the proposed block is not committed, speculative state may be discarded or rolled back.

The broader idea of executing before final agreement has substantial precedent in distributed systems~\cite{kotla2007zyzzyva,wester2009tolerating,nopaxos,speculative-paxos,fabric}.
These systems show that execution need not always follow a previously agreed total order.
For our architecture, execute-before-agree is fundamental:
multiple dependent cross-shard activations can execute within one agreement interval only if execution advances before their schedule is globally agreed.

The challenge is validating an execution-generated schedule that unfolds causally across independently executing shards,
without trusting the replica that produced it.
Prior work addresses parts of this problem.
Zyzzyva supports speculative execution under Byzantine faults, but replicas execute an order first proposed by the primary~\cite{kotla2007zyzzyva}.
Rex is closer to \gen{}: execution determines a causal trace that is later agreed upon and replayed, but Rex assumes crash faults~\cite{guo2014rex}.
\gen{} brings this execute-generated trace model into a Byzantine setting where execution itself may generate new cross-shard work.
Followers must therefore validate both the recorded execution order and the provenance of cross-shard continuations.
Crucially, this validation must not reintroduce cross-shard coordination.

The key to validating such an execution without restoring cross-shard coordination is that validity need not be established after every intermediate activation.
It is sufficient for the execution represented by a block to satisfy the system's validity conditions at the block boundary.
Because execution is deterministic given the initial account state and per-shard execution order,
follower shards can replay their traces independently, without maintaining a globally synchronized intermediate execution prefix or reproducing the leader's real-time continuation arrival order.
Moreover, validation need not delay ongoing execution:
shards may continue executing speculatively while earlier blocks are being validated and finalized.

The result is that consensus no longer sits between dependent activations:
a continuation may execute as soon as it becomes available at its destination,
and that activation may immediately generate further work.
Multiple causally dependent cross-shard activations can therefore execute within one agreement interval instead of requiring a new agreement interval for every causal step -- a property we call \emph{Causal Compression}.
Subject to available execution time,
a causal chain spanning multiple shards can execute and commit within a single block rather than incur block latency proportional to its causal depth.
At the same time, independent accounts execute concurrently across shards. 
The architecture therefore supports both parallel execution and low-latency causal computation without requiring agreement between every cross-shard step.

Our contributions are:

\begin{itemize}

    \item \textbf{Causal execute-before-agree:}
    a Byzantine-replicated sharded architecture that executes causally generated cross-shard work before agreement and validates the resulting schedule independently.

    \item \textbf{Causal Compression:}
    dependent cross-shard chains can execute within one agreement interval rather than one interval per causal step.

    \item \textbf{Causal Liveness:}
    a formal progress guarantee for continuations generated by committed computation.
    
    \item \textbf{Evaluation:} \gen{} scales execution to 1.07M requests/s
and reduces 16-hop causal-chain latency by 8.2$\times$; on NEAR testnet, every measured causal hop required an additional agreement round.

\end{itemize}

%% file: model.tex
\section{Model}
\label{sec:model}

\paragraph{System and fault model.}
\gen{} consists of $n$ clusters.
Each cluster contains $m$ shards and one consensus representative (Figure~\ref{fig:gensmr}).
\gen{} uses a black-box Byzantine SMR protocol to commit global blocks across clusters.
Every cluster stores a complete replica of the application state.
All clusters use the same fixed partition of that state among their shards,
and this partition is known to every cluster.
Thus, \gen{} shards execution within each replicated cluster; replication itself remains across complete clusters.
At any time, one cluster is designated as the \emph{leader} and the remaining $n-1$ clusters act as \emph{followers}.

The cluster, typically operated as a single administrative unit, is the unit of failure.
Shards within a correct cluster trust one another and communicate without consensus.
If any shard of a cluster is compromised, we consider the entire cluster Byzantine.
The adversary may control at most $f<n/3$ clusters.
Accordingly, \gen{} assumes no separate fault threshold for individual shards.

Communication is partially synchronous.
After an unknown global stabilization time (GST), inter-cluster and intra-cluster message delays are bounded.
We let $\delta_{\mathit{exec}}$ bound the time for a correct follower cluster to process a global block, and $\delta_{\mathit{intra}}$ bound intra-cluster communication delay.

\paragraph{Execution model.}
In \gen{}, application state is partitioned among a set of accounts $\mathcal{A}$,
each of which owns disjoint state and is assigned to one shard.
An account is therefore never split across shards, and a shard's size is the combined state of the
accounts assigned to it; we treat this assignment as fixed and leave repartitioning outside our scope.
Computation proceeds through \emph{activations}.

An activation is a uniquely identified request containing a target account and an application-defined payload.
An activation may read or modify only its target account's state.
It is therefore executed by the shard to which its target account is assigned.
An activation submitted by a client is \emph{external}.
Executing an activation $A$ against its target account's state $s$ produces
an updated state $s'$ and a sequence $C$ of zero or more new activations,
called \emph{continuations}: $\mathsf{Execute}(s,A)=(s',C)$.
A continuation's identifier cryptographically binds its immutable contents;
execution-assigned metadata such as its timestamp and block identifier is excluded.
Execution is deterministic: the same pair $(s,A)$ always produces the same
pair $(s',C)$.
This property allows followers to verify a leader's proposal through re-execution.

\paragraph{Causality.}
We define $A \prec A'$ as the transitive closure of causal generation ($A \rightarrow A'$) and same-shard execution order.
This is the happens-before relation over activations.

\paragraph{Problem definition.}
For replication, executions are grouped into global blocks.
Let $k$ denote a block index.
A global block $B^k$ contains one shard block $B_i^k$ per shard $i$, where
$B_i^k$ is the ordered sequence of activations executed by shard $i$ since
its preceding block boundary.
Applied to the preceding committed state, these per-shard execution orders
determine the resulting state and generated continuations.
The shard blocks include both external activations and internally generated
continuations. Because concurrently generated continuations may arrive in
different orders at different clusters, replicas cannot derive a common
per-shard order from their local arrival orders. The per-shard execution
orders are therefore part of the global block agreed upon through SMR.

\gen{} implements state machine replication and must satisfy:

\begin{itemize}
    \item \textbf{Agreement:} all correct replicas commit the same sequence
    of blocks in the same order.
    \item \textbf{Validity:} every committed block satisfies the protocol's
    validity predicate.
    \item \textbf{Liveness:} every activation submitted by a correct client
    eventually executes.
\end{itemize}

As is standard, Liveness is eventual: the leader controls mempool selection
within its view, so time-bounded inclusion is a property of the underlying
protocol's leader rotation rather than of the execution layer.
Cross-shard execution creates an additional obligation.
Standard Liveness covers requests submitted by correct clients.
It does not imply progress for continuations, which are generated internally
and may cross shard boundaries.
Without an explicit guarantee, a committed activation tree could therefore
stall after any cross-shard step even though standard Liveness continues to hold.
We therefore define \emph{Causal Liveness}, which extends progress to
internally generated continuations:

\begin{definition}[Causal Liveness]
If a committed activation $A$ produces a continuation $A'$, then $A'$
eventually executes.
\end{definition}

%% file: cluster_new.tex
\section{The \gen{} Architecture}
\label{sec:architecture}

\begin{figure}[t]
    \centering

    \begin{minipage}[t]{0.65\textwidth}
        \centering
        \includegraphics[width=\linewidth]{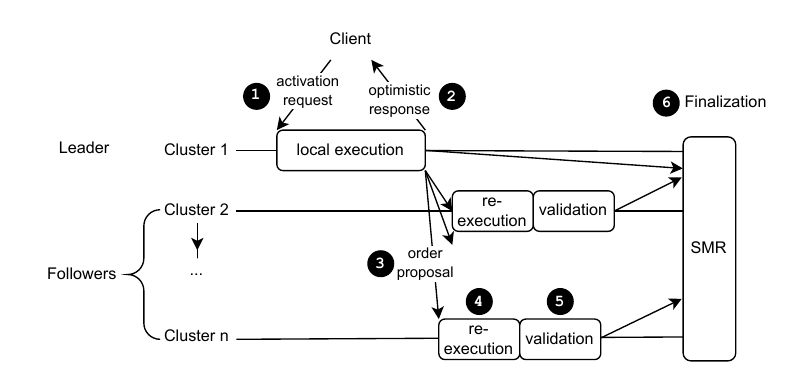}
        \par\smallskip
        \textbf{(a) Execution workflow}
    \end{minipage}
    \hfill
    \begin{minipage}[t]{0.34\textwidth}
        \centering
        \includegraphics[width=\linewidth]{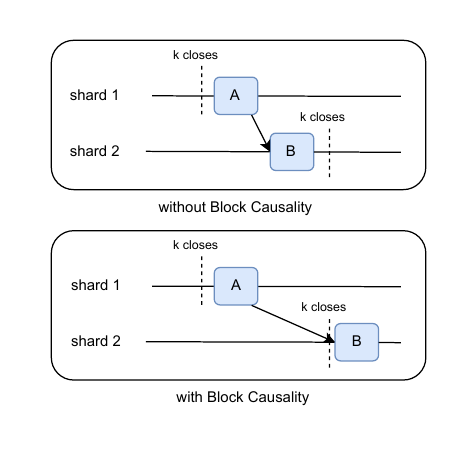}
        \par\smallskip
        \textbf{(b) Block Causality}
    \end{minipage}

    \caption{\gen{} execution workflow and Block Causality.
    }

    \label{fig:execution-and-causality}

\vspace{-0.2em}
\begin{minipage}{0.97\textwidth}
\small

\textbf{(a)}
A client submits an external activation to the leader \circled{1}, where the
target shard executes it immediately and may generate cross-shard continuations
without waiting for consensus.
The leader may return an optimistic result \circled{2}.
At the block boundary, each shard's recorded execution order forms part of a
global execution-trace proposal sent to follower clusters \circled{3}.
Followers re-execute the proposed traces against their local state \circled{4}
and verify authenticity, continuation generation, and causal dependencies
\circled{5}; shards verify independently and need not reproduce the leader's
continuation arrival order.
Once accepted, the cluster participates in inter-cluster SMR, which commits a
global block sequence \circled{6}.
Clusters whose speculative state diverges from the committed block roll back
and apply the committed execution.
\textbf{(b)}
Block Causality ensures that activations appear in monotonically ordered blocks
with respect to their causal dependencies.
\end{minipage}
\end{figure}

Building on Section~\ref{sec:model}, we first specify the per-shard protocol
state, its invariants, and leader execution
(Sections~\ref{subsection:data_structures}--\ref{subsection:leader_protocol});
we then describe block construction and replication
(Section~\ref{subsection:blocks_and_replication}), followed by follower
verification (Section~\ref{subsection:follower_protocol}). Leader and follower
shards run Algorithms~\ref{alg:shard_at_leader}
and~\ref{alg:shard_at_follower} respectively, over the shared structures of
Algorithm~\ref{alg:util_alg}\algapp.
Figure~\ref{fig:execution-and-causality}a summarizes the end-to-end execution and validation pipeline.

\subsection{Per-Shard State}
\label{subsection:data_structures}

Each activation carries protocol metadata:
$sid$ identifies its destination shard, $ts$ is its Lamport timestamp,
and $is\_external$ and $signature$ distinguish and authenticate client submissions.
A \textbf{ShardBlock} packages one shard's ordered execution trace for replication.
The corresponding data structures appear in Algorithm~\ref{alg:util_alg}\algapp.

Each shard $i$ maintains the following ordered lists:
\begin{description}
    \item [\textbf{$mempool_i$}] An ordered list of activations pending execution.
    \item [\textbf{$executed_i$}] An ordered list of activations executed by
    the shard, including both finalized activations and activations awaiting
    SMR finalization.
\end{description}

Activations enter the mempool from two sources: external activations submitted by clients,
 and continuations spawned by previously executed activations.
We make the following assumption about the external activations.

\begin{assumption}
    \label{assumption:external_activation}
    Every external activation in the leader's mempool is correctly signed by a client,
    and each external activation appears at most once in
    its associated shard's mempool
     (duplicate submissions from clients are filtered before insertion).
\end{assumption}

\subsection{Execution Model and Correctness Properties}
\label{subsection:system_state_invariants}

Shards execute pending activations concurrently, while each shard serializes its own execution.
Formally, an execution at the leader consists of per-shard sequences $\{\sigma_i\}_{i=1}^m$ where $\sigma_i = \sigma_i^0 \sigma_i^1 \ldots$ represents steps taken by shard $i$,
each being either an ApplyActivation invocation (Algorithm~\ref{alg:shard_at_leader}) or adding an activation to mempool (Algorithm~\ref{alg:shard_at_leader}).
Shards execute concurrently;
 steps within each $\sigma_i$ are totally ordered, but steps across different shards are only partially ordered by causal dependencies.

We model ApplyActivation as atomic: it acquires a lock on shard $i$ before executing and releases it afterward,
ensuring each shard processes at most one activation at a time.
Adding an activation to the mempool is lock-free and can occur concurrently with execution.
The relative ordering between a mempool addition and a concurrent ApplyActivation is determined by which operation accesses
the mempool first during execution.
This model ensures the execution history is equivalent to some sequential execution.

The system state consists of $(executed_i, mempool_i)$ for all shards $i \in \{1, \ldots, m\}$, together with all account state and shard metadata (e.g., local timestamps).

\paragraph{System State Invariants}
We now formalize what constitutes a valid system state.
These invariants define valid execution at the leader and later (Section~\ref{subsection:follower_protocol}) 
serve as the validity predicate for follower verification.
A valid system state satisfies the following properties:

\begin{description}
    \item [\textbf{Authenticity}] Every external activation in $executed_i$ is signed correctly by a client.
    \item [\textbf{Integrity}] For every continuation $A'$ in $executed_i$ or $mempool_i$, there exists an activation $A \in executed_j$ such that $A \rightarrow A'$.
    \item [\textbf{No Duplication}] Each activation appears at most once in $executed_i \cup mempool_i$.
    \item [\textbf{Well-formedness}] Every activation is included in its associated shard.
    \item [\textbf{Happens-before Relation}] For any two activations $A$ and $A'$ in $executed_i$ and $executed_j$ respectively, if $A\prec A'$, then $A.ts < A'.ts$.
\end{description}

Note that Authenticity and Integrity are dual properties that validate activation origins:
Authenticity verifies external activations come from legitimate clients,
 while Integrity verifies continuations were correctly generated by prior executions.

\subsection{Leader Protocol}
\label{subsection:leader_protocol}

The leader per-shard protocol is described in Algorithms~\ref{alg:util_alg} and~\ref{alg:shard_at_leader}\algapp.
Its core is \textsc{ApplyActivation}, which repeatedly selects and processes
eligible activations from the mempool.
The leader repeatedly selects and processes an activation from the mempool, subject to dependencies we formalize later.
The leader first validates external activations by checking their signatures.
Next, it assigns a Lamport timestamp to the activation by taking the maximum of the shard's current timestamp
and the activation's timestamp (0 for external activations), 
then incrementing by one. This ensures the happens-before relation is preserved.
The activation is then added to the $executed$ list and executes.
The execute function performs two critical roles: it runs the activation's logic to generate any follow-up continuations, and it updates the system state
according to the writes produced during execution.
Each generated continuation inherits the timestamp of its parent activation, maintaining causal ordering across shards.
Finally, the generated continuations are sent to their destination shards.
When a shard receives a continuation, it adds it to its mempool, making it available for future execution.

This protocol ensures that activations are executed in a causally consistent order while 
maintaining the integrity and correctness properties of the system state,
 which we formally prove in Appendix~\ref{appendix:proofs}:

\begin{restatable}{theorem}{statevaliditytheorem}\label{theorem:validity_of_system_state}
    Algorithm~\ref{alg:shard_at_leader} preserves a valid system state.
\end{restatable}
  
\paragraph{Leader Replacement.}

When the leader cluster fails or becomes unresponsive, a new leader cluster is selected to continue execution.
The new leader resumes from the system state that all correct clusters maintain (proven in Appendix~\ref{sec:gen_implements_smr}),
which includes all continuations generated during execution.

External activations submitted to the failed leader but not yet executed are not included in this shared state.
Clients detect leader failure through timeouts and resubmit pending activations to the new leader.
Assumption~\ref{assumption:external_activation} prevents duplicate execution by filtering resubmitted activations that already appear in the system state;
this assumption is straightforward to enforce in practice through standard deduplication mechanisms.
In \gen{}, leader selection is coupled with the underlying SMR protocol.
 The cluster currently leading the consensus protocol
  (e.g., PBFT leader~\cite{castro1999practical})
   serves as the execution leader.
 View changes in the consensus protocol trigger execution leader changes,
  with the new leader initializing from some agreed upon checkpoint.

\subsection{Blocks and Replication}
\label{subsection:blocks_and_replication}

The leader continuously executes activations, whereas replication proceeds in
discrete blocks.
Each shard therefore divides its local execution history into ordered intervals.
At the end of an interval, \textsc{CloseBlock} packages the activations executed
since the preceding boundary as a shard block, sends it to the corresponding
follower shards and the leader's consensus representative, and advances the local
block identifier $bid$.
Each executed activation records the identifier $act.bid$ of the interval in which it executed.

\textsc{ApplyActivation} and \textsc{CloseBlock} synchronize their access to
$bid$; for clarity, the pseudocode models both operations as atomic.
\textsc{CloseBlock} performs only bookkeeping: it does not execute an activation
or modify account state, $executed_i$, or $mempool_i$.
Thus, adding \textsc{CloseBlock} as a protocol step preserves the system state
invariants established above.

\paragraph{Valid Proposals.}
The leader sends each shard's block to followers as a \emph{shard proposal}:
the ordered activations executed by that shard during the block interval.
A \emph{global proposal} for block $k$ collects one shard proposal from each
of the $m$ shards. It must maintain system correctness when applied. Formally:

\begin{definition}[Valid Global Proposal]
A global proposal is \emph{valid} if applying all shard proposals to a correct system state results in another correct system state, maintaining all system state invariants.
\end{definition}

Since causal dependencies cross shards, naively batching independently formed
shard proposals may violate validity. We therefore introduce \emph{Block
Causality}, which requires activations to appear in monotonically ordered
blocks with respect to their dependencies:

\noindent\textbf{Block Causality.}
For any two activations $A$ and $A'$ in $executed_i$ and $executed_j$,
respectively, if $A \prec A'$, then $A.bid \leq A'.bid$.

Block Causality constrains how independently closed shard intervals align
(Figure~\ref{fig:execution-and-causality}b). Suppose shard 1 has advanced to block
$k+1$ while shard 2 is still constructing block $k$. If activation $a$ now
executes on shard 1 and generates continuation $a'$ for shard 2, executing
$a'$ in block $k$ would place it before its parent. The leader therefore sets
$a'.min\_bid=a.bid=k+1$, and shard 2 cannot select $a'$ until it reaches that
block.

The bound does not impose an unnecessary block barrier: when a continuation
reaches its destination before that shard closes the parent's block, it may
execute in the same block. Repeating this across shards allows a causal chain
to execute and commit in one global block, a property we call \emph{Causal
Compression}.

The minimum block identifier preserves Block Causality, with proofs deferred to
Appendix~\ref{appendix:blocks_and_replication_proofs}

\begin{restatable}{lemma}{blockcausality}\label{lemma:leader_maintains_block_causality}
    Algorithm~\ref{alg:shard_at_leader} maintains Block Causality.
\end{restatable}

With Block Causality established, we can show that the leader generates valid proposals:

\begin{restatable}{lemma}{correctleadercreatesvalidproposal}\label{lemma:correct_leader_creates_valid_proposal}
    A correct leader executing CloseBlock at all shards creates a valid global proposal.
\end{restatable}

\subsection{Follower Protocol}
\label{subsection:follower_protocol}

In the previous subsection, blocks partitioned the leader's continuous execution into discrete units for replication.
Blocks also define the checkpoints at which followers validate execution before participating in consensus.
Follower shards receive the leader's block proposals and re-execute them locally,
forwarding a block to the cluster's consensus representative only if verification succeeds.
The pseudocode appears in Algorithms~\ref{alg:shard_at_follower} and \ref{alg:shard_at_follower_verification} (Appendix~\ref{appendix:algorithms}).

\paragraph{Follower State}

A follower, like the leader, maintains $executed$ and $mempool$ fields for each shard.
However, followers process blocks atomically -- either the entire block is applied, or none of its activations are.
During verification, \emph{speculated} and \emph{pending\_activations} hold the activations that would be added to \emph{executed} and \emph{mempool}.
If verification succeeds, these shadow structures are merged into permanent state (making the block application atomic).
Otherwise, they are discarded.
Followers also maintain \emph{ts\_at\_execution}, which records the local timestamp at which each continuation is executed
 and is used to verify the leader's Lamport timestamps.

Followers also maintain block-boundary checkpoints of
\emph{executed}, \emph{mempool}, the local timestamp, and account state.
If a block is not committed, the consensus representative restores the preceding
checkpoint (Section~\ref{sec:gen_consensus}), allowing the follower to safely
retry under a new proposal.

\subsubsection*{Protocol Overview}

When a follower receives a shard proposal for block $k$,
it first ensures that blocks are processed sequentially, buffering the proposal if $k \neq bid$.
The follower then invokes \textsc{IsValidBlock}, which verifies the proposal in two phases.

In the first phase, the follower speculatively re-executes the proposed activations in order.
For each activation, it checks shard assignment and the Block Causality constraint $min\_bid \leq bid$.
External activations can be verified immediately by checking signatures, duplication, and Lamport timestamp computation, since they have no dependencies
on other shards.
For continuations, the follower executes them speculatively, recording the local timestamp at execution and sending any generated continuations to their destination
shards within the cluster.
However, continuations cannot yet be fully verified: verification depends on receiving matching continuations from other shards.

In the second phase, the follower waits until $\delta_{\mathit{exec}}+\delta_{\mathit{intra}}$ from the start of block
processing, then verifies cross-shard continuations.
The verification function then validates that the leader's reported continuations match those generated locally.
For each continuation, it checks that a matching locally generated continuation exists in \emph{pending\_activations} or \emph{mempool}, and that its Lamport
timestamp was computed correctly.

If verification succeeds, the follower merges \emph{speculated} into \emph{executed} and \emph{pending\_activations} into \emph{mempool}, then
forwards the shard block to the consensus representative.
Otherwise, it reports the block as invalid and discards the speculative state.

This protocol yields an important correctness property: followers need not enforce cross-shard validity after each activation.
Activations on different shards may be re-executed in parallel, so a continuation may be processed before its parent completes on another follower
shard. 
The protocol checks cross-shard validity only at the block boundary, after re-executing the entire proposal and allowing its continuations to arrive.
Shadow state and atomic block application ensure that these temporary inconsistencies never become visible.

Appendix~\ref{appendix:follower_proofs} proves that successful block verification preserves the system state invariants
(Lemma~\ref{lemma:block_execution_correctness}) and that correct followers accept a correct leader's proposal under the stated timing bounds (Lemma~\ref{lemma:follower_completes_correct_leader_block}).

%% file: smr.tex
\section{State Machine Replication}
\label{sec:smr}

Traditional SMR systems replicate a single state machine across nodes, limiting scalability.
\gen{} implements SMR with horizontal scalability through sharding:
each replica is a cluster of shards that execute in parallel.
To coordinate between clusters, \gen{} uses an existing SMR protocol as a building block.
This creates a two-level structure:
sharded execution within clusters, and SMR-based consensus between clusters.

\subsection{The Underlying SMR Protocol}
\label{sec:smr_building_block}

\gen{} uses a standard SMR protocol as a black-box component to order global blocks and ensure agreement among clusters.
Each cluster has a consensus representative responsible for participating in the SMR protocol with representatives from other clusters.

For liveness, we assume the standard partial-synchrony progress condition of the underlying SMR protocol: after GST, its leader-replacement mechanism eventually
installs a correct leader that remains active long enough to complete a proposal.
At block height $k$, the representative submits the global block $B^k=(k,[B_1^k,\ldots,B_m^k]),
$
containing one verified shard block from each of the $m$ shards.

However, \gen's usage of SMR differs from traditional systems in several ways.
First, \gen{} uses execute-before-agree rather than commit-then-execute:
blocks are executed speculatively before consensus,
and SMR commit serves as after-the-fact validation.
The SMR ``execution'' step simply persists pre-computed state to disk.
Second, the SMR clients are the execution leader (system nodes), not end users.
End users interact with the execution layer via the mempool.
Third, and most significantly, the replicated global block is an execution trace rather than merely a batch of external requests.
Each shard block records the leader-selected order in which both external activations and internally generated continuations executed.
Followers replay and validate this proposed order instead of deriving an order from local continuation arrival times, which may differ across clusters.
The underlying SMR protocol orders and commits a sequence of global blocks; it does not reorder activations within their shard blocks.
This differs from traditional SMR, where ordered requests originate externally from clients.
\gen's validity predicate is whether executing a block maintains all system state invariants --
formally, whether the IsValidBlock function returns true at all shards.
Commit certificates for committed blocks are those of the underlying protocol, inherited unchanged
for use by light clients and external verifiers.

\subsection{Consensus Representative Protocol}
\label{sec:gen_consensus}

The consensus representative coordinates between a cluster's shards and the inter-cluster SMR protocol
(Appendix~\ref{appendix:smr_proofs}, Algorithm~\ref{alg:consensus_representative_alg}).
When shards complete block execution and verification, they send their results (valid or invalid) to the consensus representative.
The representative collects shard blocks with the same block id, binds them together to form the global block $B^k$, and proposes it to SMR.
The validity predicate for SMR is whether all shards within the cluster accepted their respective shard blocks --
that is, whether IsValidBlock returned true at all shards.
A correct representative votes only for a global block whose shard blocks
exactly match those accepted by its local shards.
If any shard rejected its block, the representative does not propose to SMR, allowing the underlying consensus protocol to handle leader replacement through its standard fault tolerance mechanisms.

Once SMR commits a global block, there are two cases.
If the committed block matches the block that the cluster executed speculatively, the representative persists the pre-computed state to disk, creating a checkpoint.
Otherwise, if the committed block differs, the representative triggers rollback to the previous checkpoint, re-executes the committed block, and then persists the resulting state.
This could be the case, for example, if the leader is Byzantine and sends a different proposal to a single follower.
If SMR rejects a block (due to insufficient agreement, timeout, or view change),
the representative coordinates rollback across all shards in the cluster,
restoring state from the checkpoint at the previous block.

The specific choices of SMR protocol and persistence mechanism (logging, snapshots, etc.) are orthogonal to \gen's core design.
Similarly, in practice consensus would operate on compact commitments (e.g., Merkle roots of shard block hashes) rather than full block data, requiring additional construction and verification steps, but we omit this standard optimization for clarity.
Appendix~\ref{appendix:smr_proofs} proves that these mechanisms satisfy the properties stated in the Model:

\begin{restatable}{theorem}{genimpsmr}\label{theorem:gen_implements_smr}
 \gen{} satisfies Agreement, Validity, Liveness, and Causal Liveness.
 \end{restatable}

%% file: perf_tight.tex
\section{Performance and Design Implications}
\label{sec:perf}

The execute-before-agree architecture has three important consequences for
system performance.

\textbf{Cross-Shard Execution.}
Because a cluster's shards are co-located, dependent continuations travel over
the intra-cluster execution path rather than through consensus.
Causal chains can therefore progress at execution speed, with agreement applied
to the resulting execution trace rather than between dependent steps.

\textbf{Leader-Follower Parity.}
Execute-before-agree is useful only if followers can validate execution at a
rate comparable to the leader.
\gen{} enables this by replaying per-shard traces independently and checking
cross-shard dependencies at block boundaries, leaving execution across shards
parallel.

\textbf{Decoupling Execution from Block Timing.}
Because blocks record completed execution rather than prescribe work that must
finish within a block interval, execution need not be synchronized with block
boundaries.
This permits both optimistic responses before finality and activations whose
execution duration need not fit within a block interval.

Further discussion of communication, speculative execution, long-running activations,
and follower verification appears in Appendices~\ref{appendix:additional_perf}
and~\ref{appendix:design_discussions}.

%% file: eval_new.tex
\vspace{-0.5em}

\section{Evaluation}
\label{sec:eval}

We evaluate \gen{} along three dimensions: horizontal execution scaling,
the latency benefit of Causal Compression over deferred execution,
and recovery from speculative execution after leader failure.
We additionally validate the deferred-execution model against NEAR testnet.

Our implementation comprises approximately 460K lines of Rust across execution,
consensus, networking, storage, and runtime components.
Experiments run on AWS Graviton instances in a single region with approximately
1\,ms inter-validator RTT, isolating execution and replication costs rather than
wide-area consensus latency.
External activations and consensus messages are Ed25519-signed;
generated continuations are validated against their source execution.
Agreement operates over a Merkle root of shard-block hashes, keeping the
consensus payload constant-size.

We use YCSB~\cite{cooper2010benchmarking} for scalability and fault-tolerance
benchmarks and a fungible token transfer workload for cross-shard benchmarks,
where each transfer debits the sender and generates a continuation that credits
the receiver.
Each validator dedicates one core to the consensus representative; the remaining
cores host execution shards.

\paragraph{Horizontal scalability.}

We first measure whether throughput scales with available execution parallelism.
With a single validator, \gen{} scales from 16 to 160 execution cores and reaches
approximately $1.07$M requests/s.
Across six measured configurations, per-core throughput remains nearly constant,
closely tracking linear scaling over a $10\times$ increase in compute resources
(Figure~\ref{fig:eval_scaling}a).

We next repeat the experiment with four validators, so proposed blocks are
independently re-executed and verified before agreement.
Scaling each validator from 4 to 48 cores 
raises replicated YCSB throughput from $16{,}200$ to $193{,}000$
committed requests/s: a $11.9\times$ gain over a $12\times$ increase in compute
(Figure~\ref{fig:eval_scaling}b).
The fungible token workload scales $9.3\times$ over the same range, reaching
approximately $65$K committed transfers/s.
Because receivers are chosen uniformly at random, increasing the shard count
also increases the fraction of transfers that cross shards, from $67\%$ at
3 shards to $98\%$ at 47.
Despite becoming increasingly cross-shard, throughput remains close to linear.
Thus, replicated throughput continues to scale with execution parallelism.

\begin{figure}[t!]
    \centering
    \begin{minipage}{0.49\linewidth}
        \centering
        \includegraphics[width=\linewidth]
        {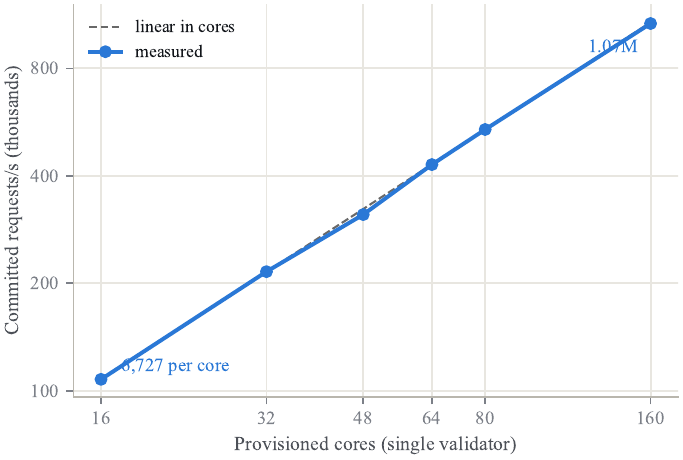}
        \vspace{-1ex}
        \subcaption{Single-validator execution.}
    \end{minipage}
    \hfill
    \begin{minipage}{0.49\linewidth}
        \centering
        \includegraphics[width=\linewidth]
        {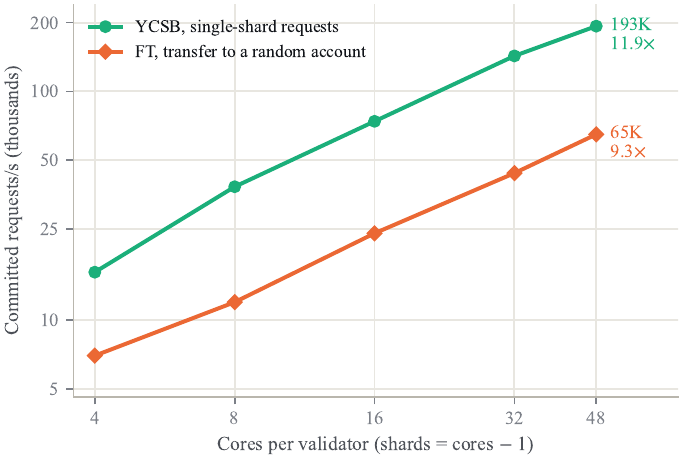}
        \vspace{-1ex}
        \subcaption{Four-validator replication.}
    \end{minipage}
    \caption{Horizontal scalability of \gen{}.
    (a) Single-validator execution.
    (b) Four-validator replication under YCSB and fungible-token workloads.}
    \label{fig:eval_scaling}
\end{figure}

\paragraph{Leader--Follower Parity.}
Execute-before-agree is useful only if followers can verify at the leader's rate.
Across 24 runs of the cross-shard workload at 44--100\% of sustainable load,
the slowest follower shows no persistent block-height lag.
The leader-to-slowest-follower gap has an average fitted growth rate of
$0.030$ blocks/s (standard deviation $0.212$) and decreases in 10 of 24 runs.
Because insufficient follower capacity would appear as steadily increasing lag,
these results indicate that follower validation keeps pace with leader execution.

\paragraph{Causal Compression.}

We evaluate chains of $N$ sequential continuations against an otherwise identical configuration that defers each continuation to the next block.
This isolates the cost of an agreement boundary between dependent steps while keeping the execution engine, workload, network stack, and deployment unchanged.
As shown in Figure~\ref{fig:eval_causal_failure}a, deferred execution adds approximately $103$\,ms per causal step against a $100$\,ms block interval,
 while \gen{} adds only approximately $4$\,ms.
At 16 hops, median latency is $1,822$\,ms deferred versus $222$\,ms under \gen{}, an $8.2\times$ reduction versus deferred execution.

\emph{Real-system validation.}
To verify that this deferred-execution pattern occurs in practice,
 we measured causal chains on NEAR testnet.
With 10 active shards, every measured causal hop across depths 2--16 and 80 runs required one additional agreement round, including cross-shard chains.
TON's basechain operated as a single shard during our measurements, so it could not provide a cross-shard baseline.

\paragraph{Fault tolerance.}
Finally, we test whether speculative execution recovers after leader failure.
We kill the active leader during sustained YCSB load in a four-validator
deployment.
Across three runs at $2{,}000$ requests/s, surviving validators recover after
the configured view-change timeouts and return to their pre-failure throughput,
as shown in Figure~\ref{fig:eval_causal_failure}b.
The transient spike reflects queued requests being processed after recovery.
The experiment validates the recovery path rather than peak-load failover.

\begin{figure}[t]
    \centering
    \begin{minipage}{0.49\linewidth}
        \centering
        \includegraphics[width=\linewidth]
        {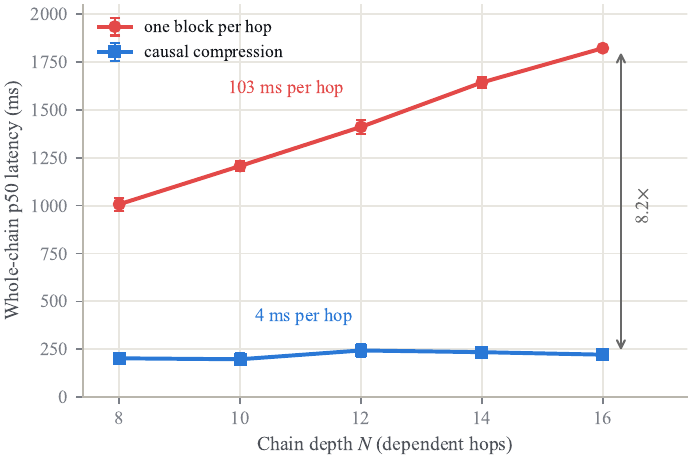}
        \vspace{-1ex}
        \subcaption{Causal Compression.}
    \end{minipage}
    \hfill
    \begin{minipage}{0.49\linewidth}
        \centering
        \includegraphics[width=\linewidth]
        {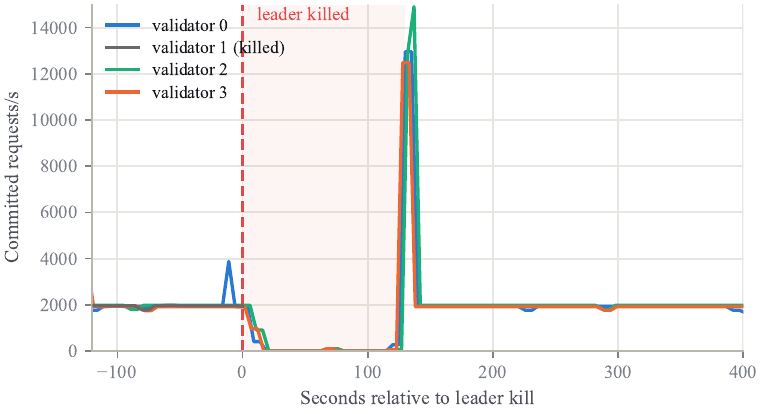}
        \vspace{-1ex}
        \subcaption{Leader-failure recovery.}
    \end{minipage}
    \caption{
        (a) Causal Compression avoids per-hop block delay.
        (b) Validators recover after leader failure and resume pre-failure throughput.
        Error bars in (a) show one standard deviation over three runs.
    }
    \label{fig:eval_causal_failure}
\end{figure}

%% file: related_tight.tex
\vspace{-0.5em}
\section{Related Work}

\textbf{Blockchain Scalability Approaches.}
Modern L1 blockchains scale either vertically, through parallel execution on a single chain, or horizontally, through sharding.
Aptos~\cite{aptos}, Sui~\cite{sui}, and Solana~\cite{solana} follow the former approach, but conflicting or shared-state transactions still require coordination.
NEAR~\cite{near}, TON~\cite{ton}, Dfinity/ICP~\cite{icp}, QuarkChain~\cite{quarkchain}, and Monoxide~\cite{monoxide} shard state and execution.
\gen{}, like NEAR and TON, retains a single consensus point, but lets dependent cross-shard work progress without the per-step agreement
boundaries observed in NEAR's receipt-based execution.

\textbf{Cross-shard semantics.}
Prior sharded systems either restrict cross-shard transactions to operations known in advance~\cite{rapidchain}
or execute data-dependent cross-shard work through successive asynchronous receipts~\cite{near,monoxide,quarkchain}.
\gen{} instead localizes atomicity to an activation and lets applications introduce stronger coordination only where needed.
\gen{} and NEAR provide asynchronous composability in the sense used by CIRC~\cite{espresso-circ};
CIRC achieves this across rollups through a centralized coordinator, whereas \gen{} provides the property within a Byzantine-replicated sharded L1.
This is related to broader work on weakening system-wide ACID semantics in partitioned systems~\cite{cure,cops,calvin,occult},
but \gen{} additionally guarantees Causal Liveness for internally generated continuations.

\textbf{Execute-before-agree.}
\gen{} builds on prior execute-before-agree systems including Speculative Paxos~\cite{speculative-paxos}, NOPaxos~\cite{nopaxos}, Hyperledger Fabric~\cite{fabric}, and Rex~\cite{guo2014rex}.
Unlike these systems, \gen{} must validate Byzantine-proposed traces in which execution dynamically generates further cross-shard operations,
while allowing follower shards to verify independently and defer cross-shard consistency checks to the block boundary.

%% file: appendix/appendix_algs.tex
\section{Protocol Pseudocode}
\label{appendix:algorithms}

This appendix collects the pseudocode of the \gen{} shard protocols described in
Section~\ref{sec:architecture}: the shared data structures and common logic
(Algorithm~\ref{alg:util_alg}), the leader shard protocol
(Algorithm~\ref{alg:shard_at_leader}), and the follower shard protocol
(Algorithms~\ref{alg:shard_at_follower}
and~\ref{alg:shard_at_follower_verification}).
Line numbers referenced from the body of the paper refer to these algorithms.

\input{algs/util_alg.tex}

\input{algs/cluster_leader_alg.tex}

\input{algs/cluster_alg.tex}

%% file: algs/util_alg.tex
\begin{algorithm}[H]
	\caption{\gen{}: Data Structures and Common Logic} 
    \label{alg:util_alg}
	\begin{algorithmic}[1]

        \Statex \textbf{define struct} Activation:
        \begin{algorithmic}[1]
            \Statex $id$: integer \Comment{activation id}
            \Statex $logic$: function \Comment{smart contract function}
            \Statex $ts$: integer \Comment{Lamport timestamp (0 for external)} %
            \Statex $sid$: integer %
            \Statex $min\_bid$: integer \Comment{minimum block id (0 for external)} %
            \Statex $bid$: integer \Comment{block id containing this activation} %
            \Statex $is\_external$: boolean %
            \Statex $signature$: bytes \Comment{client signature (external only)}
        \end{algorithmic}

        \Statex \textbf{define struct} ShardBlock: 
        \begin{algorithmic}[1]
            \Statex $bid$: integer 
            \Statex $sid$: integer 
            \Statex $data$: List$\langle$Activation$\rangle$  \Comment{activations executed by this shard}
        \end{algorithmic}

        \Statex \textbf{define struct} GlobalBlock: 
        \begin{algorithmic}[1]
            \Statex $bid$: integer
            \Statex $blocks$:  List[m] of ShardBlock \Comment{one per shard}
        \end{algorithmic}
        \Statex

        \Function{Execute}{$activation$}
            \State $(writes, continuations) \gets$ activation.logic()
            \State ProcessWrites($writes$) \Comment{update shard's account state}
            \For{$cont \in continuations$}
                \State $cont.ts \gets activation.ts$ \label{line:set_ts_for_cont}
                \State $cont.min\_bid \gets activation.bid$ \label{line:set_min_bid_for_cont}
            \EndFor
            \State \Return continuations
        \EndFunction
        \Statex

        \Function{Select}{$activations$}
            \State \Return activation with highest priority \Comment{any policy without indefinite postponement; FIFO in base implementation}\label{line:select}
        \EndFunction

	\end{algorithmic} 
\end{algorithm}

\newcounter{tempALGline}
\setcounter{tempALGline}{\value{ALG@line}} %

%% file: algs/cluster_leader_alg.tex
\begin{algorithm}[H]
	\caption{\gen{}: Leader Shard $i$ Protocol} 
    \label{alg:shard_at_leader}
	\begin{algorithmic}[1]

        \setcounter{ALG@line}{\value{tempALGline}} %

        \Statex \textbf{Local variables}:
        \Statex $bid$: integer, initialized to 0
        \Statex $ts$: integer, initialized to 0
        \Statex $mempool$: list⟨Activation⟩, initialized to []

        \Statex $executed$: list$\langle$Activation$\rangle$, initialized to [] 
        \Statex $order\_proposal$: list$\langle$Activation$\rangle$, initialized to [] %
        \Statex

        \Loop
        \State wait until $\exists a \in mempool$ such that $a.min\_bid \leq bid$
        \State \Call{ApplyActivation}{} \label{line:apply_activation_at_leader}
        \EndLoop
        \Statex

        \Upon{clock tick} 
        \State \Call{CloseBlock}{} \Comment{Does not interrupt ApplyActivation}
        \EndUpon
        \Statex

        \Upon{receipt of continuation \emph{cont}}
            \State $mempool$.append(\emph{cont}) \label{line:add_cont_to_pending_activations}
        \EndUpon
        \Statex

        \Procedure{ApplyActivation}{}

            \State $filtered \gets [\,a \in mempool \mid a.min\_bid \leq bid\,]$~\label{line:filter_pending_activations}
            \State $\emph{act} \gets$ \Call{Select}{$filtered$}
            \State $mempool.\text{remove}(\emph{act})$~\label{line:remove_act_from_pending_activations}

            \If{\emph{act}.is\_external and \emph{act}.signature is invalid}~\label{line:leader_checks_external_activation}

                \State return
            \EndIf

            \State \emph{act}.ts $\gets \max(ts, \emph{act}.ts)+1$ \Comment{Lamport ts calculation} \label{line:leader_sets_ts_for_act}  %
            \State $ts \gets \emph{act}.ts$ \label{line:leader_updates_ts}

            \State \emph{act}.bid $\gets bid$ \label{line:set_bid_for_act}
            \State $executed$.append(\emph{act}) \label{line:add_act_to_order_proposal}
            \State $order\_proposal$.append(\emph{act})

            \State continuations  $\leftarrow$  \Call{Execute}{$\emph{act}$} \label{line:execute_act_at_leader}

            \ForAll{$cont \in continuations$}
                \State send $cont$ to destination shard $cont.shard$ \label{line:send_cont_to_shard}
            \EndForAll
       \EndProcedure
        \Statex

        \Procedure{CloseBlock}{}\label{line:leader_closes_block}

            \State send $\langle{order\_proposal, bid }\rangle$ to shard $i$ at all followers \label{line:leader_sends_block_to_followers} %
                
            \State $block \gets$ ShardBlock($bid, i, order\_proposal$) \label{line:leader_creates_block}
            \State send $block$ to leader consensus representative %
            \State $bid \gets bid + 1$ \label{line:increment_bid_at_leader} 
            \State $order\_proposal \gets []$ \label{line:leader_clears_order_proposal}

        \EndProcedure

	\end{algorithmic} 
\end{algorithm}
\setcounter{tempALGline}{\value{ALG@line}} %

%% file: algs/cluster_alg.tex
\begin{algorithm}[H]
	\caption{\gen{}: Follower Shard $i$ Protocol}
    \label{alg:shard_at_follower}
	\begin{algorithmic}[1]
        \setcounter{ALG@line}{\value{tempALGline}} %

        \Statex \textbf{Local variables}:
            \Statex $bid$: integer, initialized to 0
            \Statex $ts$: integer, initialized to 0

            \Statex $mempool$: list$\langle$Activation$\rangle$, initialized to [] 
            \Statex $executed$: list$\langle$Activation$\rangle$, initialized to [] 
            \Statex $pending\_activations$: list$\langle$Activation$\rangle$, initialized to []
            \Statex $speculated$: list$\langle$Activation$\rangle$, initialized to []

            \Statex $ts\_at\_execution$: map$\langle$Activation, integer$\rangle$, initialized to $\emptyset$
        \Statex

        \Upon{receipt of continuation \emph{cont}}
            \State $pending\_activations$.append(\emph{cont}) \label{line:add_cont_to_pending_activations_follower}

        \EndUpon
        \Statex

        \Upon{receipt of $\langle{order\_proposal}, k \rangle$ from leader shard $i$}
            \If{$k \neq bid$}
                \State buffer block until $bid = k$ \Comment{Execute block proposals sequentially}
            \EndIf
            \State $is\_valid \gets$ \Call{IsValidBlock}{$order\_proposal$}
            \If{$is\_valid$}

            \State $block \gets$ ShardBlock($bid, i, order\_proposal$)
            \State send $block$ to consensus representative \label{line:follower_sends_block_to_consensus_representative} %
            \State $bid \gets bid + 1$
            \State $executed \gets executed + speculated$ \label{line:add_speculated_to_executed}
            \State $mempool \gets mempool + pending\_activations$\label{line:combine_mempool_and_pending_activations} 

            \Else
                \State send $\langle$\texttt{invalid}, $block\rangle$ to consensus representative \label{line:follower_sends_invalid_block_to_consensus_representative}

            \EndIf
            \State $pending\_activations \gets []$
            \State $speculated \gets []$ \label{line:end_of_block_proposal_execution}
            \State $ts\_at\_execution$ $\gets \emptyset$ \label{line:clear_ts_at_execution_at_follower}
        \EndUpon

	\end{algorithmic} 
\end{algorithm}
\setcounter{tempALGline}{\value{ALG@line}} %

\begin{algorithm}[H]
	\caption{\gen{}: Follower Shard $i$ Protocol (continued)}
    \label{alg:shard_at_follower_verification}
	\begin{algorithmic}[1]

        \setcounter{ALG@line}{\value{tempALGline}} %

        \Function{IsValidBlock}{$order\_proposal$}

            \State $start\_time \gets$ current time \label{line:set_start_time_for_phase_1}
            \Statex \textbf{Phase 1: Re-execution}

            \For{$\emph{act} \in order\_proposal$}

            \If{$\emph{act}.shard \neq i$ or $\emph{act}.min\_bid > bid$} \label{line:return_false_if_act_not_valid}
                \State \Return false \label{line:return_false_if_act_not_valid_at_follower_shard_or_bid}
            \EndIf
            \If{\emph{act}.is\_external}
                \State $ts \gets ts + 1$ \label{line:increment_ts_at_follower_for_external_act}
                \If{$\emph{act}.ts \neq ts$ or $\emph{act}$ is not properly signed or $\emph{act}$ is in $executed_i$ or \emph{act} is in $speculated_i$} \label{line:return_false_if_act_not_valid_at_follower_signed_external}
                    \State \Return false \label{line:return_false_if_act_not_valid_at_follower_signed_external_or_in_executed_or_speculated}
                \EndIf
            \Else \Comment{this is a continuation}
                \State $ts\_at\_execution$[\emph{act}] $\gets ts$ \label{line:set_ts_at_execution_for_act}
                \State $ts \gets \emph{act}.ts$ \label{line:set_ts_for_cont_at_follower}
            \EndIf

            \State $speculated$.append(\emph{act}) \label{line:add_act_to_speculated_at_follower}
            \State continuations  $\leftarrow$  \Call{Execute}{$\emph{act}$} \label{line:execute_act_at_follower}

            \ForAll{$cont \in continuations$}
                \State send $cont$ to destination shard $cont.shard$ \label{line:send_cont_to_shard_at_follower}
            \EndForAll

        \EndFor

        \Statex \textbf{Phase 2: Verification}
        \State $deadline \gets start\_time + \delta_{exec} + \delta_{intra}$
        \State wait until current time $\geq$ deadline \label{line:wait_for_deadline_at_follower}

        \State \Return \Call{IsCausallyConsistent}{$speculated$}

        \EndFunction
        \Statex

        \Function{IsCausallyConsistent}{$speculated$}

            \ForAll{$cont \in speculated$} 

                \If{not $cont.is\_external$} \Comment{Reported continuations need to be verified against internal computation}
                \State $in\_cont \gets \bot$

                    \If{$cont.id \in pending\_activations$}

                    \State $in\_cont \gets pending\_activations$.remove($cont.id$)\label{line:remove_cont_from_pending_activations} 
                    \ElsIf{$cont.id \in mempool$}
                        \State $in\_cont \gets mempool$.remove($cont.id$) \label{line:remove_cont_from_mempool}
                    \EndIf

                    \If{not $in\_cont$ or $cont.ts \neq max(in\_cont.ts, ts\_at\_execution[in\_cont])+1$}
                        \State \Return false \label{line:return_false_if_cont_not_found_or_ts_not_correct}
                    \EndIf
                \EndIf

            \EndForAll

            \State \Return true
        \EndFunction

	\end{algorithmic} 
\end{algorithm}

%% file: appendix/appendix.tex
\section{Correctness Proofs}
\subsection{Leader Protocol Proofs}
\label{appendix:proofs}

This appendix contains proofs of lemmas and theorems from the main text, which were omitted due to space constraints.

\statevaliditytheorem*
\begin{proof}
    Our execution model ensures any concurrent execution is equivalent to a sequential one, ordered by lock acquisition.
    We prove correctness by induction on the number of execution steps.
    The base case is when the system is empty, which is trivially valid.
    Assume that the system state is valid after $n$ steps, and consider the next step.

    \textbf{Adding an activation to the mempool:}
    Consider a step where an activation is added to some shard's mempool.
    Without loss of generality, let this be $mempool_i$.
    Activations are added to the $mempool_i$ either externally,
    in which case they satisfy Assumption~\ref{assumption:external_activation},
    or as a continuation \emph{cont} at line~\ref{line:add_cont_to_pending_activations}.
    In the latter case, \emph{cont} is sent by the leader at line~\ref{line:send_cont_to_shard} after executing activation \emph{act} at shard $j$.
    By line~\ref{line:execute_act_at_leader} \emph{cont} was generated according to the activation logic.
    Then, at line~\ref{line:add_act_to_order_proposal}, \emph{act} was added to $executed_j$ before sending the continuation to shard $i$, preserving Integrity.
    By the code, and by the induction hypothesis on \emph{act} satisfying No Duplication, this addition only occurs once and at the correct shard, preserving No Duplication and Well-formedness.
    The other properties are unaffected by adding an activation to the mempool.

    \textbf{Applying an activation:}
    
    Consider a step where some shard executes an activation.
    Without loss of generality, let this be shard $i$ executing activation \emph{act}.

    For \textit{Authenticity}, the leader only executes an external activation and adds it to $executed_i$ after verifying its signature is valid (line~\ref{line:leader_checks_external_activation}).

    Since \emph{act} is taken from $mempool_i$ (line~\ref{line:filter_pending_activations}), it was added to $mempool_i$ in some earlier step.
    By the induction hypothesis, all invariants held after that earlier step, including:
    
    \textit{Integrity:} If \emph{act} is a continuation, it was correctly generated by executing some previous activation in $executed_j$ (for some shard $j$).
    Moving \emph{act} from $mempool_i$ to $executed_i$ preserves this property-- the same activation that was correctly generated remains correctly generated.
    
    \textit{No Duplication:} The activation \emph{act} appeared at most once in $mempool_i \cup executed_i$.
    Since \emph{act} is removed from $mempool_i$ (line~\ref{line:remove_act_from_pending_activations}) before being added to $executed_i$ (line~\ref{line:add_act_to_order_proposal}), \emph{act} still appears at most once.
    
    \textit{Well-formedness:} The activation \emph{act} is included in shard $i$ when added to $mempool_i$.
    Since it remains on the same shard throughout execution, it is added to the correct $executed_i$.

    Finally, we prove the \textit{Happens-before Relation} is preserved by our Lamport timestamps implementation.
    Let $A$ and $A'$ be two activations in $executed_i$ and $executed_j$ respectively such that $A \prec A'$.
    We consider the cases defining $A \prec A'$.

    If $i = j$, then $A$ occurs before $A'$ in $executed_i$.
    Let $A.ts$ be the timestamp assigned to $A$ at line~\ref{line:leader_sets_ts_for_act}.
    Within each shard, the local timestamp $ts$ is incremented monotonically with every activation execution (line~\ref{line:leader_updates_ts}).
    Thus, when $A'$ is executed, its timestamp is set to at least $A.ts + 1$ (line~\ref{line:leader_sets_ts_for_act}), ensuring $A.ts < A'.ts$.

    If $i \neq j$ and $A$ directly triggered $A'$,
    then when $A'$ is generated as a continuation of $A$, it inherits $A$'s timestamp: $A'.ts = A.ts$ (line~\ref{line:set_ts_for_cont}).
    When $A'$ is later executed on shard $j$, its timestamp is updated to at least $A.ts + 1$ (line~\ref{line:leader_sets_ts_for_act}), ensuring $A.ts < A'.ts$.

    If $A$ did not directly trigger $A'$, there exists a causal chain of activations from $A$ to $A'$.
    Let $A''$ be the activation that directly triggered $A'$ ($A'' \rightarrow A' $ and $A \prec \ldots \prec A'' \prec  A'$.).
    By the arguments above, $A''.ts < A'.ts$.
    By the induction hypothesis applied to the earlier step when $A''$ was added to $executed_j$, we have $A.ts < A''.ts$.
    By transitivity, $A.ts < A'.ts$.

\end{proof}

\subsection{Blocks and Replication Proofs}
\label{appendix:blocks_and_replication_proofs}

\blockcausality*
\begin{proof}
Causal relations within a shard are preserved naturally by the sequential execution order: if $A$ and $A'$ are both in shard $i$ and $A\prec A'$,
then $A$ is executed before $A'$, so $A.bid \leq A'.bid$.

For cross-shard causality, when activation $A$ in shard $i$ directly triggers continuation $A'$ in shard $j$ (i.e., $A \rightarrow A'$), we set $A'.min\_bid = A.bid$ (line~\ref{line:set_min_bid_for_cont}).
The leader's selection of activations respects this constraint by filtering for activations where $activation.min\_bid \leq bid$ (line~\ref{line:filter_pending_activations}).
Therefore, $A'$ can only be included in a block with $bid \geq A.bid$, ensuring $A.bid \leq A'.bid$.

For transitive causality where $A \prec A'$ through a chain of activations,
the result follows by induction on the length of the causal chain, using the two base cases above.
\end{proof}

Having established that the leader maintains Block Causality, we can now prove that the global proposals generated by the leader are valid.

\correctleadercreatesvalidproposal*
\begin{proof}
By Theorem~\ref{theorem:validity_of_system_state}, a correct leader maintains a correct system state after every execution step.
When all shards execute CloseBlock for block $k$, the resulting global proposal consists of all activations added to $executed_i$ across all shards $i$ since block $k-1$ was closed.

The Authenticity, No Duplication, and Well-formedness properties are maintained within each shard independently and do not involve cross-shard dependencies. Since every prefix of a correct leader's execution maintains these properties (by Theorem~\ref{theorem:validity_of_system_state}), they hold for the activations in each shard proposal.

The Integrity and Happens-before Relation properties require that causal dependencies are preserved in the global proposal.
That is, if $A\prec A'$ and $A'$ is included in the global proposal, then $A$ must either already be in $executed_j$ (for some shard $j$) before $A'$ is executed,
or $A$ must also be included in the global proposal.
This is guaranteed by Lemma~\ref{lemma:leader_maintains_block_causality} (Block Causality) and the fact that $bid$ increments monotonically: if $A.bid \leq A'.bid$ and both are in blocks $\leq k$, then either $A.bid < A'.bid$ (so $A$ was in an earlier block) or $A.bid = A'.bid$ (so both are in the same global proposal).

Therefore, applying the global proposal to the system state at the end of block $k-1$ results in a correct system state at the end of block $k$.
\end{proof}

\begin{remark}
Note that Authenticity, No Duplication, and Well-formedness are maintained within each 
shard independently and do not involve cross-shard dependencies. 
This means that a shard proposal can be marked as invalid with regard to these properties
 during follower verification based solely on the local execution history. 
 We rely on this fact in subsequent proofs.
\end{remark}

\subsection{Follower Proofs}
\label{appendix:follower_proofs}

\subsubsection*{Follower Protocol Correctness}

We now establish the correctness of the follower protocol.
We first prove that execution is deterministic (Lemma~\ref{lemma:deterministic_execution}),
which ensures followers can reliably verify leader proposals through re-execution.
We then prove that successful block execution maintains system state invariants (Lemma~\ref{lemma:block_execution_correctness}).
Finally, we prove that when the leader is correct and timing assumptions hold, followers successfully complete block verification (Lemma~\ref{lemma:follower_completes_correct_leader_block}).

\paragraph{Execution Determinism}

\begin{lemma}
\label{lemma:deterministic_execution}
Given the same initial state and the same shard proposal, executing the activations sequentially in the order they appear produces the same final state.
\end{lemma}
\begin{proof}
Each activation's execution is deterministic: given the same state and activation, the execute function produces the same writes and continuations.
By induction on the number of activations in the shard proposal, executing them sequentially in order produces the same final state.
\end{proof}

\paragraph{Block Execution Safety}

The previous lemma establishes that if a follower successfully executes a proposal, it reaches the same state as the leader (by deterministic execution from the same initial state).
However, it does not address whether the proposal itself is valid—that is, whether it maintains system state invariants.

A key challenge in our sharded architecture is handling cross-shard causality during parallel execution.
Consider activations $A$ and $A'$ on different shards where $A \rightarrow A'$.
During follower execution, different shards process their shard proposals in parallel:
$A'$ may execute on its shard before $A$ completes on its shard, creating temporary ``inconsistencies'' in the system state.

Our approach is to tolerate these transient inconsistencies during execution, but enforce that all invariants hold at block boundaries.
By checking system state correctness only after all shards complete their execution (rather than coordinating during execution),
we enable shards to execute independently in parallel, maximizing throughput.
This block-boundary verification strategy is central to \gen's scalability:
 shards execute at full speed without inter-shard coordination, synchronizing only at block boundaries.
The use of shadow data structures and atomic application of blocks ensures that these transient 
 inconsistencies are never visible to the external clients.

With this approach established, we must prove that the verification procedure correctly identifies invalid proposals.
Specifically, we show: if IsValidBlock returns true for all shards, then applying the global proposal maintains all system state invariants.
For the proof, we observe that our invariants partition into two categories.
Shard-local properties (Authenticity, No Duplication, Well-formedness) can be verified by examining only a single shard's execution history.
Cross-shard properties (Integrity, Happens-before Relation) require checking that causal dependencies between shards are properly maintained.

We introduce terminology for block execution outcomes.
A shard $i$ \emph{completes block $k$ successfully} if it executes and verifies the shard block $k$ proposal, then approves it for consensus (line~\ref{line:follower_sends_block_to_consensus_representative}).
Otherwise, if verification fails and the shard sends an invalid block message (line~\ref{line:follower_sends_invalid_block_to_consensus_representative}), the shard has \emph{failed block $k$}.

We now prove that successful block execution at all shards maintains system state invariants,
beginning with auxiliary lemmas about how follower state evolves during and after block execution.

We begin by proving that continuations in $pending\_activations$ have valid sources in $executed$.

\begin{lemma}
    \label{lemma:pending_activations_implies_src_in_executed}
    Let $b$ be the block number during which activation $A$ is added to $pending\_activations_i$.
    If block $b$ completes successfully at all shards, then there exists another activation $B$ such that $B\rightarrow A$,
    and $B$ is added to $executed_j$ by the end of block $b$.
    Moreover, each instance of $A$ in $pending\_activations_i$ corresponds to a distinct instance of $B$ in $executed_j$.
\end{lemma}
\begin{proof}
    If $A$ is found in $pending\_activations_i$ it means by the code that it was added to $pending\_activations_i$ at line~\ref{line:add_cont_to_pending_activations_follower}, after being sent by some shard $j$ at Line~\ref{line:send_cont_to_shard_at_follower}.
    The sending event of $A$ happens after the execution of the activation $B$ at line~\ref{line:execute_act_at_follower}, and after it was added to $speculated_j$ at line~\ref{line:add_act_to_speculated_at_follower}
    during some block processing.
    This execution is done according to the logic ($B\rightarrow A$) and creates all follow-up continuations, including $A$, with min block id that is $B$'s block id (line~\ref{line:set_min_bid_for_cont}).
    Thus, since $A$ previously passed the check at the beginning of block $b$ execution in line~\ref{line:return_false_if_act_not_valid},
     it is guaranteed that $B$ has been executed in block $b' \leq b$.
    When shard $j$ finishes executing block $b'$, it appends all activations in $speculated_j$ (including $B$) to $executed_j$ at
     line~\ref{line:add_speculated_to_executed}.
    If IsValidBlock returned false earlier, block $b'$ fails at shard $j$.
    In the first case, since $b'\leq b$ and $executed_j$ is append-only, $B$ is guaranteed to be added to $executed_j$ by the end of block $b$.
    Otherwise, since blocks are processed sequentially, block $b$ that includes $A$ does not complete.
    Note that by the code, the instance of $A$ in $pending\_activations_i$ corresponds to a distinct instance of $B$ in $speculated_j$ as after $B$'s execution it is only created once, and only added once to the pending set.
    
\end{proof}

The next lemma extends this result to continuations that have been promoted from $pending\_activations$ to $mempool$.

\begin{lemma}
    \label{lemma:mempool_implies_src_in_executed}
    Let $b$ be the block number when activation $A$ is added to $mempool_i$.
    If block $b$ completes successfully at all shards, and $A$ is a continuation, then there exists another activation $B$ such that $B\rightarrow A$,
    and $B$ is added to $executed_j$ by the end of block $b$.
    Moreover, each instance of $A$ in $executed_i$ corresponds to a distinct instance of $B$ in $executed_j$.
\end{lemma}
\begin{proof}
    By the code, $A$ is added to $mempool_i$ only if it is found in $pending\_activations_i$ (line~\ref{line:combine_mempool_and_pending_activations}).
    The proof then follows from Lemma~\ref{lemma:pending_activations_implies_src_in_executed}.
\end{proof}

We now prove the same property for continuations that have been fully verified and added to $executed$,
completing the chain from $pending\_activations$ through $mempool$ to $executed$.

\begin{lemma}
    \label{lemma:executed_implies_src_in_executed}
    Let $b$ be the block number when activation $A$ is added to $executed_i$.
    If block $b$ completes successfully at all shards, and $A$ is a continuation, then there exists another activation $B$ such that  $B\rightarrow A$,
    and $B$ is added to $executed_j$ by the end of block $b$.
    Moreover, each instance of $A$ in $executed_i$ corresponds to a distinct instance of $B$ in $executed_j$.
  \end{lemma}
  \begin{proof}
    By the code, $A$ is added to $executed_i$ from $speculated_i$ at line~\ref{line:add_speculated_to_executed}, after IsValidBlock returned true for block $b$.
    This happens after all continuations in $speculated_i$ were found in either $mempool_i$ or $pending\_activations_i$ 
    (lines~\ref{line:remove_cont_from_pending_activations}--\ref{line:return_false_if_cont_not_found_or_ts_not_correct}).
    If $A$ appears in $pending\_activations_i$, proof follows from Lemma~\ref{lemma:pending_activations_implies_src_in_executed}.
    Otherwise, $A$ must be found in $mempool_i$ and proof follows from Lemma~\ref{lemma:mempool_implies_src_in_executed}.
  \end{proof}

These three lemmas together prove the Integrity property: every continuation in $executed_i$, $mempool_i$ has a valid source.
We now use these results to prove that No Duplication is maintained throughout block execution.

  \begin{lemma}[No Duplication]
    \label{lemma:no_duplication}
    For every $k$, at the end of executing block $k$ at all shards, there does not exist an activation that appears in $executed_i \cup mempool_i$ more than once, for any $i$.
  \end{lemma}
  \begin{proof}
    Define the causal depth of an external activation to be zero and, whenever
    $A \rightarrow A'$, define $\mathit{depth}(A')=\mathit{depth}(A)+1$.
    We prove this lemma with induction on $k$.
    
    Base case: $k=0$. All data structures are empty.

    Inductive step: 
    Assume that the statement is true for any $i\leq k$ and prove for $i=k+1$.
    Assume by negation the statement does not hold.
    Thus, there exists a shard $i$ and an activation that appears in $executed_i$ and $mempool_i$ more than once, after block $k+1$ has completed at all shards.
    Let $A$ be such activation, with the shortest causal depth $d$. 
    If $A$ is an external activation (i.e., $d=0$), then $A$ is not added to follower's $mempool_i$ and thus it does not appear in $mempool_i$ more than once.
    Hence, it must appear in $executed_i$ more than once.
    By the code, $A$ is added to $executed_i$ only if it was previously added to $speculated_i$ at line~\ref{line:add_act_to_speculated_at_follower}.
    This only happens if it passes the check at line~\ref{line:return_false_if_act_not_valid_at_follower_signed_external}, making sure it was not already in $executed_i$ or $speculated_i$.
    Hence, the second addition of $A$ to $executed_i$ is a contradiction, and the statement holds.

    Otherwise, $A$ is a continuation.
    We separate the cases that cause the violation.
    \begin{itemize}
        \item $A$ appears more than once in $mempool_i$:
        
        Examine two different occurrences of $A$ in $mempool_i$.
        By Lemma~\ref{lemma:mempool_implies_src_in_executed} applied twice per each occurrence, since we know by assumption all shards have completed block $k+1$,
        there exists another activation $B$ such that $B\rightarrow A$, and $B$ is added twice to $executed_j$ by the end of block $k+1$.
        This means $B$ appears in $executed_j$ more than once, after block $k+1$ has completed at all shards,
        and $B$ has causal depth $d-1$, in contradiction to the assumption that $A$ has the shortest causal depth.

        \item $A$ appears more than once in $executed_i$:
        Similar to the previous case, examine two different occurrences of $A$ in $executed_i$.
        By Lemma~\ref{lemma:executed_implies_src_in_executed} applied twice per each occurrence, since we know by assumption all shards have completed block $k+1$,
        there exists another activation $B$ such that $B\rightarrow A$, and $B$ is added twice to $executed_j$ by the end of block $k+1$.  
        This means $B$ appears in $executed_j$ more than once, after block $k+1$ has completed at all shards,
        and $B$ has causal depth $d-1$, in contradiction to the assumption that $A$ has the shortest causal depth.

        \item $A$ appears in both $mempool_i$ and $executed_i$:
        Finally, examine the two occurrences of $A$ in $mempool_i$ and $executed_i$.
        By Lemma~\ref{lemma:mempool_implies_src_in_executed} applied for the occurrence in $mempool_i$, and Lemma~\ref{lemma:executed_implies_src_in_executed} applied for the occurrence
        in $executed_i$, since we know by assumption all shards have completed block $k+1$, there exists another activation $B$ such that $B\rightarrow A$, and $B$ is added to $executed_j$ twice
        by the end of block $k+1$.
        This means $B$ appears in $executed_j$ more than once, after block $k+1$ has completed at all shards, and $B$ has causal depth $d-1$, in contradiction to the assumption that $A$ has the shortest causal depth.

    \end{itemize}
  \end{proof}

Having proved that No Duplication holds, we now turn to the Happens-before Relation property.
We first show that local timestamps at each shard increase monotonically during successful block execution.

\begin{lemma}
    \label{lemma:monotonicity_of_local_ts}
    If shard $i$ completes the execution of block $k$ successfully,
    then the local timestamp $ts$ in shard $i$ is strictly monotonically increasing throughout the execution up to and including block $k$.
\end{lemma}

\begin{proof}
Assume by negation that the local timestamp in shard $i$ is not strictly monotonically increasing.
The local timestamp changes when processing an external activation (line~\ref{line:increment_ts_at_follower_for_external_act})
or when processing a continuation (line~\ref{line:set_ts_for_cont_at_follower}).
In both cases, the change corresponds to processing an activation and adding it to $speculated_i$.

By assumption, there exist two activations $A$ and $A'$ such that $A$ is processed before $A'$ (either within the same block or across blocks),
 with shard's local timestamps $t$ and $t'$ respectively at the time they are processed, such that $t \geq t'$.
Choose the first such pair -- that is, $A'$ is the first activation processed after $A$ where $t \geq t'$.

If $A'$ is external, then by the code, the local timestamp is set to $t' = t + 1$ (line~\ref{line:increment_ts_at_follower_for_external_act}),
 contradicting $t' \leq t$.

Otherwise, $A'$ is a continuation.
When $A'$ is processed, the follower records $ts\_at\_execution[A'] \gets t$ (line~\ref{line:set_ts_at_execution_for_act}), 
then sets the local timestamp to $t' = A'.ts$ (line~\ref{line:set_ts_for_cont_at_follower}).
Later, when the verify function verifies $A'$, it checks that $A'.ts \geq ts\_at\_execution[A'] + 1 = t + 1$ (line~\ref{line:return_false_if_cont_not_found_or_ts_not_correct}), or otherwise fails.
Therefore $t' \geq t + 1$, contradicting the assumption that $t' \leq t$.

Since block $k$ completes successfully, no verification fails, so our assumption must be false.
Therefore the local timestamp is strictly monotonically increasing throughout the execution.
\end{proof}

With monotonicity of local timestamps proven,
we can now show that the Happens-before Relation is preserved for all activation pairs.

\begin{lemma}[Happens-before Relation]
    \label{lemma:happens_before_relation}
        If all shards complete the execution of block $k$ successfully, then for any two activations $A, A'$ in $executed_i$ and $executed_j$ respectively,
        if $A\prec A'$, then $A.ts < A'.ts$.
    \end{lemma}
    \begin{proof}
        Let $A$ and $A'$ be two activations in $executed_i$ and $executed_j$ respectively such that $A\prec A'$.
        We consider the cases that define $A\prec A'$.
    
        Assume $i = j$ and $A$ appears before $A'$ in the sequential order of $executed_i$.
        At the end of block $k$, $executed_i$ is appended with $speculated_i$.
        Since the order in $speculated_i$ reflects the order in which activations were processed during IsValidBlock,
        if $A$ was processed in block $k$, then $A$ was processed before $A'$.
        Note that this ordering holds even if $A$ was executed in an earlier block, 
        since activations are appended to $executed_i$ in order across blocks.

        Let $t$ be the local timestamp in shard $i$ when $A$ was executed.
        By Lemma~\ref{lemma:monotonicity_of_local_ts}, since shard $i$ completes blocks successfully,
        the local timestamp is strictly monotonically increasing both within and across blocks.
        Therefore, the local timestamp when processing $A'$ in block $k$ is strictly greater than the 
        local timestamp when $A$ was processed (whether in block $k$ or an earlier block).
        This means that when $A'$ is executed, the local timestamp in shard $i$ is at least $t+1$.

        Since $A'$ is added to $executed_i$, it passed verification-- either the check at 
        line~\ref{line:return_false_if_act_not_valid_at_follower_signed_external} (if external) or the check 
        at line~\ref{line:return_false_if_cont_not_found_or_ts_not_correct} (if continuation).
        In both it was compared to the local timestamp (either explicitly or to a recorded value captured in $ts\_at\_execution[A']$).
        and $A'.ts$ is verified to be at least $t+1$.
        Therefore, $A'.ts > A.ts$.

        Alternatively, if $i \neq j$ and $A$ triggered $A'$, then when the follower verified $A'$, the verify function checked
         that $A'.ts \geq incoming\_cont.ts + 1$, where $incoming\_cont$ is the locally generated continuation corresponding to $A'$ 
         (line~\ref{line:return_false_if_cont_not_found_or_ts_not_correct}).
        When the follower speculatively executed $A$, it set $incoming\_cont.ts = A.ts$ (line~\ref{line:set_ts_for_cont}).
        Therefore, $A'.ts \geq A.ts + 1 > A.ts$.
        
        The transitive case, where $A \prec A'$ through a chain of intermediate activations, 
        follows by induction on the length of the causal chain.
    
        Therefore, in all cases, $A.ts < A'.ts$.

    \end{proof}

The preceding lemmas verify individual properties during block execution.
We now combine these results to prove that successful block execution maintains all system state invariants.

\begin{lemma}
\label{lemma:block_execution_correctness} 
    If all shards complete the execution of block $k$ successfully, then applying the global proposal for block $k$ maintains all system state invariants.
\end{lemma}
\begin{proof}

Unlike the leader's execution where all properties are maintained after every step,
at followers some properties are only maintained at block boundaries -- that is, at the end of block execution across all shards.
    
Our invariants partition by verification scope.
\textbf{Authenticity} and \textbf{Well-formedness} are verified locally during each activation's processing
(lines~\ref{line:return_false_if_act_not_valid_at_follower_shard_or_bid}, \ref{line:return_false_if_act_not_valid_at_follower_signed_external_or_in_executed_or_speculated}).
Failing to satisfy these properties for any activation immediately causes IsValidBlock to return false at the relevant shard, 
and the shard fails the execution of block $k$.

\textbf{Integrity} and \textbf{Happens-before Relation} cannot be verified by examining individual activations in isolation -- they require reasoning about causal relationships across shards and may be temporarily violated during parallel execution.
\textbf{No Duplication} similarly requires reasoning across a shard's entire execution history, though it is never violated within a single shard during execution.
These three properties are verified at block boundaries through the preceding lemmas:
Integrity follows from Lemmas~\ref{lemma:mempool_implies_src_in_executed} and~\ref{lemma:executed_implies_src_in_executed},
No Duplication from Lemma~\ref{lemma:no_duplication},
and Happens-before Relation from Lemma~\ref{lemma:happens_before_relation}.
Therefore, if all shards complete block $k$ successfully, all system state invariants hold.
\end{proof}

\paragraph{Block Execution Liveness}

Our previous focus was on showing that successfully completing block $k$ implies the proposal is valid (safety).
We now turn to liveness: when the leader is correct and conditions are favorable,
all correct followers successfully complete block execution.
We begin by showing that timestamp verification failures indicate leader misbehavior.

\begin{lemma}
    \label{lemma:follower_verifies_leader_ts}
    If a follower with the same state as the leader before block $k$ fails a timestamp check for activation $A$
    (at line~\ref{line:return_false_if_act_not_valid_at_follower_signed_external_or_in_executed_or_speculated} or line~\ref{line:return_false_if_cont_not_found_or_ts_not_correct}),
    then the leader did not correctly apply the Lamport timestamp mechanism for $A$.
    \end{lemma}
    \begin{proof}
    The follower verifies timestamps by recomputing them according to the Lamport timestamp rules and comparing to the assigned ts by the leader:
    incrementing for external activations, and computing $\max(parent.ts, local\_ts) + 1$ for continuations.
    Since the follower starts with the same state and executes the same deterministic logic,
    it computes the same timestamp values the leader should have computed.
    If verification fails, the leader's reported timestamp differs from the correct Lamport timestamp.
\end{proof}

With timestamp verification correctness established, we can now prove the main liveness result:
correct leaders produce proposals that correct followers accept.
This is the first lemma that reasons about both leader and follower behavior together,
showing that correct leader execution guarantees successful follower verification.

\begin{lemma}
    \label{lemma:follower_completes_correct_leader_block}
    Assume the leader is correct and sends a proposal for block $k$.
    After GST, a correct follower that receives the proposal and has the same state as the leader had before executing block $k$
    successfully completes execution of block $k$ at all shards.
\end{lemma}
\begin{proof}

We prove the contrapositive.
Assume a correct follower receives a proposal for block $k$ after GST and has the same state as the leader had before executing block $k$,
but fails to complete execution of block $k$ at some shard $i$.
Then the leader is not correct.

By the code, the execution failure means that \texttt{IsValidBlock} called on $order\_proposal$ returns false at shard $i$.
We examine all reasons that could cause this failure.
Note that any of the checks that may fail is associated with a specific activation in the proposal.
Denote this activation by $A$.

\begin{itemize}
    \item Line~\ref{line:return_false_if_act_not_valid_at_follower_shard_or_bid}:
    By the code, this line returns false if  either $A.shard$ is not the shard $i$ on which it is attempted to be executed
    or $A.min\_bid$ is greater than the bid of the block.
    In the first case, the well-formedness property is violated, implying the leader is not correct (Lemma~\ref{lemma:correct_leader_creates_valid_proposal}).
    In the second case, if the leader was correct, $A$ should not have been included in the proposal, implying the leader is not correct
     (Lemma~\ref{lemma:correct_leader_creates_valid_proposal}).

    \item Line~\ref{line:return_false_if_act_not_valid_at_follower_signed_external_or_in_executed_or_speculated}:
    By the code, this line returns false if $A$ is external activation and one of the following holds.

    If $A.ts$ is not equal to the timestamp $ts$ of the shard $i$ at the time of execution.
    By Lemma~\ref{lemma:follower_verifies_leader_ts} and Lemma~\ref{lemma:correct_leader_creates_valid_proposal}, the leader must be Byzantine.

    If $A$ is not properly signed, implying the leader is not correct. (Lemma~\ref{lemma:correct_leader_creates_valid_proposal})

    If  $\emph{A}$ is already in $executed_i$ or $speculated_i$ before it is executed, 
    it follows from the code that $A$ was previously executed on this shard.
    Thus, the leader's proposal violates the No Duplication property, implying the leader is incorrect. (Lemma~\ref{lemma:correct_leader_creates_valid_proposal})

    \item Line~\ref{line:return_false_if_cont_not_found_or_ts_not_correct}:
    By the code, this line returns false if $A$ is a continuation and one of two conditions holds:
    (1) $A$ is not found in either $pending\_activations_i$ or $mempool_i$, or
    (2) $A$ is found but its timestamp does not satisfy the Lamport timestamp constraint.

    \textbf{Case 1: $A$ is missing.}

    Since the leader creates correct proposals (Lemma~\ref{lemma:correct_leader_creates_valid_proposal}),
    there must be another activation $B$ such that $B \rightarrow A$ and $B \in executed_j$ when the proposal was created.
    Due to Block Causality, $B$ was executed at the leader in some block $k' \leq k$ on shard $j$.

    If $k' < k$, then by timing bounds after GST, $A$ arrives at the leader's shard $i$ within $\delta_{intra}$ time
    (which is smaller than block tick interval), so $A$ is added to either $mempool_i$ by the end 
    of block $k'$ or during block $k' + 1$.
    
    If $A$ reaches $mempool_i$ by the end of any block $k'' < k$, then by the same-state assumption,
    $A$ is also in the follower's $mempool_i$ at the start of block $k$, contradicting the assumption 
    that $A$ is missing.
    
    Otherwise, $A$ arrives during block $k' + 1 = k$.
    Then $A$ is added to $pending\_activations_i$ before the deadline $start\_time + \delta_{exec} + \delta_{intra}$. The same occurs at the 
    follower: $B$ executes, generates $A$, and by timing bounds, $A$ arrives at 
    $pending\_activations_i$ before the deadline, contradicting the assumption that $A$ is missing.
    Hence, in all subcases when $k' < k$, the leader must be incorrect.

    Otherwise, $k' = k$. In this case, $B$ executes during block $k$ at both leader and follower.
    When the follower executes $B$, it generates $A$ and sends it to shard $i$ (line~\ref{line:send_cont_to_shard_at_follower}).
    Let $start\_time$ denote the time when the follower begins executing the block proposal.
    By the code, the follower waits until $deadline = start\_time + \delta_{exec} + \delta_{intra}$ before verifying continuations (line~\ref{line:wait_for_deadline_at_follower}).
    Since the follower receives the proposal after GST, timing bounds hold:
    $\delta_{exec}$ bounds the execution time for all activations in the proposal,
    and $\delta_{intra}$ bounds intra-cluster message delay.
    Therefore, $A$ arrives at shard $i$ and is added to $pending\_activations_i$ before the deadline,
    contradicting the assumption that $A$ is missing.

    Hence, in both subcases, the leader must be incorrect.

    \textbf{Case 2: $A$ is found but timestamp is incorrect.}

    By Lemma~\ref{lemma:follower_verifies_leader_ts} and Lemma~\ref{lemma:correct_leader_creates_valid_proposal},
    the leader must be Byzantine.

    Therefore, in all cases, the leader is not correct.

\end{itemize}

\end{proof}

%% file: appendix/appendix_smr.tex
\section{\gen{} Consensus Representative Protocol and Correctness}
\label{appendix:smr_proofs}

\input{algs/consensus_rep_alg.tex}

\label{sec:gen_implements_smr}

We now prove that \gen{} (Algorithm~\ref{alg:consensus_representative_alg}), built atop a standard SMR component, satisfies all SMR properties including Causal Liveness.
We begin by establishing that correct clusters maintain consistent state across committed blocks.

\begin{lemma}
    \label{lemma:state_consistency}
    If two correct clusters commit the same sequence of blocks, they have reached the same system state.    
\end{lemma}

\begin{proof}
    Let $r_1$ and $r_2$ be two correct clusters that commit the same sequence of blocks $[b_0, b_1, \ldots, b_k]$.
    We prove by induction on $k$ that the system state committed with block $b_k$ is identical at both clusters
    (same $executed_i$, $mempool_i$, $ts_i$, and account state for all shards $i$).
    
    For the base case ($k = 0$), both clusters start with empty initial state.
    For the inductive step, assume $r_1$ and $r_2$ committed identical state with block $b_{k-1}$.

    When consensus commits block $b_k$, there are two cases at each cluster:
    
    If $b_k = gblock$ (committed block matches cluster's proposal).
    The cluster has already executed and validated $b_k$ during the proposal phase.
    It simply persists the pre-computed state.
    
    Otherwise, $b_k \neq gblock$ (committed block differs from cluster's proposal).
    The cluster triggers rollback to block $b_{k-1}$, restoring state from checkpoint,
    then executes the committed block $b_k$ from that restored state.

    In both cases, each cluster executes block $b_k$ starting from the state committed with block $b_{k-1}$.
    By deterministic execution (Lemma~\ref{lemma:deterministic_execution}), executing the same sequence of activations from the same initial state produces the same final state.
    Since both clusters begin with identical state (by induction hypothesis) and execute the same block $b_k$ (committed by consensus),
    they produce identical state at all shards.

    When both clusters commit block $b_k$, they merge their shadow structures into permanent state identically 
    (Algorithm~\ref{alg:shard_at_follower}, line~\ref{line:combine_mempool_and_pending_activations}).
    Their overall state, consisting both the merged data structures and account state, is what gets committed with block $b_k$.
    
    Note that if consensus had rejected $b_k$, both clusters would have performed rollback identically, 
    restoring from checkpoints at block $b_{k-1}$ (which are identical by induction hypothesis), 
    and would not have committed block $b_k$.
    Additionally, any execution that occurs after committing block $b_k$ belongs to later blocks and does not affect the state committed with $b_k$.
    Therefore, the system state committed with block $b_k$ is identical at both $r_1$ and $r_2$.
    
\end{proof}
    
Next, we show that activations in the mempool eventually execute -- a key building block for both Liveness and Causal Liveness.

\begin{lemma}
    \label{lemma:mempool_progress}
    If activation $A$ is in the leader's mempool and the leader continues executing, then $A$ is eventually selected and executed.
\end{lemma}
\begin{proof}
    By assumption, $A$ is in $mempool_i$ for some shard $i$ at leader cluster.
    By the leader protocol, the leader continuously selects activations from the mempool for execution (line~\ref{line:apply_activation_at_leader}).
    The selection procedure (line~\ref{line:filter_pending_activations}) filters activations where $activation.min\_bid \leq bid$, where $bid$ is the current block number.
    Since $bid$ increases monotonically with each block (line~\ref{line:increment_bid_at_leader}), eventually $bid \geq A.min\_bid$, making $A$ eligible for selection.
    
    The selection policy chooses from filtered activations in FIFO order based on arrival time.
    Since $mempool_i$ is finite and activations are continuously processed, $A$ will eventually be selected and executed.

    By the fairness requirement on $\mathsf{Select}$ (Algorithm~\ref{alg:util_alg}, line~\ref{line:select}), only finitely many
    activations may be selected ahead of $A$ once $A$ is eligible.
    Since the leader selects continuously, $A$ is therefore selected after finitely many selections.
    Once selected, $A$ is removed from the mempool (line~\ref{line:remove_act_from_pending_activations}) and executed (line~\ref{line:apply_activation_at_leader}).
\end{proof}

\begin{remark}
    While \gen's base implementation uses FIFO selection, more sophisticated policies could be explored for different fairness or efficiency goals,
    such as gas price-based prioritization or stake-weighted ordering, as long as they satisfy a fairness constraint preventing indefinite postponement of eligible activations.
\end{remark}

The lemmas above enable us to prove \gen's liveness properties.

\begin{lemma}
\label{lemma:causal_liveness}
\gen{} satisfies Liveness and Causal Liveness.
\end{lemma}
\begin{proof}    

    We first show that, for each property, the relevant activation is eventually executed by a correct leader.
    
    For Liveness, let $A$ be an external activation submitted by a client.
    The client submits $A$ to the current leader, which adds it to the appropriate mempool.
    If the leader cluster is correct and remains active, then by Lemma~\ref{lemma:mempool_progress}, $A$ is eventually selected and executed.
    Otherwise, by the client retransmission rule of Section~\ref{subsection:leader_protocol}, the client resubmits $A$ after leader replacement.
    By the liveness condition of the underlying SMR protocol, after GST leader replacement eventually installs a correct leader that remains active long enough to receive $A$.    
    Lemma~\ref{lemma:mempool_progress} then implies that $A$ is eventually selected and executed.
    
    For Causal Liveness, suppose activation $A$ commits in block $k$ and produces continuation $B$.
    All correct clusters execute block $k$ and generate the same continuation $B$ by deterministic execution (Lemma~\ref{lemma:deterministic_execution}),
    adding $B$ to the mempool of its destination shard.
    By Block Causality (Lemma~\ref{lemma:leader_maintains_block_causality}), $B.min\_bid = k$, ensuring that $B$ satisfies the selection constraint as blocks progress.

    If the current leader cluster is correct and remains active, then by Lemma~\ref{lemma:mempool_progress}, $B$ is eventually selected and executed.
    Otherwise, $B$ remains present across leader replacements: every correct cluster generated $B$ when executing committed block $k$, and by Lemma~\ref{lemma:state_consistency} correct clusters agree on the corresponding
    system state.
    Thus, whenever leadership moves to a correct cluster, that leader resumes with $B$ in the mempool.
    By the liveness condition of the underlying SMR protocol, after GST a correct leader eventually remains active long enough for Lemma~\ref{lemma:mempool_progress} to apply.
    If a speculative block containing $B$ is not committed, rollback restores the preceding checkpoint, in which $B$ remains pending, so $B$ remains available for a subsequent
    proposal.

    It remains to show that execution by a correct leader eventually commits.
    Let $X$ denote the relevant activation ($A$ for Liveness and $B$ for Causal Liveness).
    Once a correct leader executes $X$ after GST, it includes $X$ in a global proposal.
    By Lemma~\ref{lemma:follower_completes_correct_leader_block}, all correct followers successfully verify the corresponding shard proposals.
    By Lemma~\ref{lemma:correct_leader_creates_valid_proposal}, the resulting global proposal is valid.
    Liveness of the underlying SMR protocol therefore ensures that the proposal eventually commits.
    
    Hence $A$ eventually executes and commits for every external activation submitted by a correct client, and every continuation $B$ generated by a
    committed activation eventually executes and commits.
    Therefore, \gen{} satisfies both Liveness and Causal Liveness.    

\end{proof}

We can now state and prove the main result.

\genimpsmr*
\begin{proof}
\textbf{Agreement:} 
By Agreement property of the underlying SMR protocol, all correct clusters commit the same sequence of global blocks in the same order.
By Lemma~\ref{lemma:state_consistency}, if they all commit the same sequence of blocks, they all reach the same system state.
Therefore, \gen{} satisfies Agreement: all correct clusters agree on the same system state.

\textbf{Validity:} 
\gen's validity predicate is that the committed global block content,
when executed by correct clusters from the previous committed state, maintains all system state invariants (IsValidBlock returns true at all shards).
This is the validity predicate enforced by the underlying SMR protocol.
By the underlying SMR Validity property, all committed blocks satisfy this predicate.
Therefore, \gen{} satisfies Validity.

\textbf{Liveness and Causal Liveness:} Follows from Lemma~\ref{lemma:causal_liveness}.
\end{proof}

%% file: algs/consensus_rep_alg.tex
\begin{algorithm}
	\caption{\gen{}: Consensus Representative Protocol} 
        \label{alg:consensus_representative_alg}

        \begin{algorithmic}[1]

        \Statex \textbf{Local variables}:
        \Statex $block\_proposals$: Map[integer, GlobalBlock], initialized to $\emptyset$ \Comment{All blocks in creation}
        \Statex

        \Upon{receipt of $block$ from shard}
        \State $bid \gets block.bid$

        \If{$bid \notin block\_proposals$}
        \State $block\_proposals[bid] \gets GlobalBlock(bid, [])$
        \EndIf
        \State $gblock \gets block\_proposals[bid]$

            \State $gblock.blocks[block.sid] \gets block$

            \If{$|gblock.blocks| = m$} \Comment{blocks from all shards arrived}
                \State SMR.propose($gblock$)

            \EndIf
        \EndUpon
        \Statex

        \Upon{SMR commits block $b$}
        \If{$b = gblock$}
                \State persist $gblock$ and system state to disk
        \Else
                \State rollback to checkpoint at end of block $b.bid-1$ 
                \State execute all operations in committed block $b$
                \State persist $b$ and resulting system state to disk
        \EndIf
        \EndUpon
        \Statex

        \Upon{receipt of $\langle$\texttt{invalid}, $block\rangle$ from shard}
                \State abort block.bid and rollback to checkpoint at block.bid - 1
        \EndUpon
        \Statex

        \Upon{rollback required to block $k-1$}
                \For{all shards $i$}
                \State restore state of shard $i$ to checkpoint at end of block $k-1$ 
                
                \EndFor
        \EndUpon

	\end{algorithmic} 
\end{algorithm}

%% file: appendix/appendix_concise_discussion.tex
\section{Additional Performance Analysis}

\label{appendix:additional_perf}

\paragraph{Communication Architecture}

\label{appendix:communication_architecture}

Cross-shard latency is a first-order concern outside blockchains as well:
traditional distributed databases invest heavily in specialized hardware -- InfiniBand, RDMA, and custom interconnects --
to keep it low.
\gen{} reduces cross-shard communication latency through placement, by co-locating a cluster's shards in one data center
(Section~\ref{sec:perf}).

The inter-cluster consensus path, which is the primary consumer of inter-cluster bandwidth, operates at block boundaries rather than per activation.
Moreover, consensus can operate on compact commitments rather than full block data, significantly reducing message size.

\paragraph{Speculative Execution for Latency}
\label{appendix:speculative_execution_for_latency}

\gen{} employs speculative execution to provide low-latency responses to clients.
We maintain a clear separation between activation processing time and response finalization time.
Following the execute-before-agree pattern, clients receive optimistic responses from the leader immediately upon execution, potentially before the block is even constructed.
If the leader is correct, the speculative response reflects a valid execution, although it is not yet final.

Clients requiring finality guarantees can track consensus progress or wait for finality confirmation,
but many applications prioritize fast responses over guaranteed finality,
including DeFi trading platforms, blockchain gaming, and real-time interactive applications.
Together with Causal Compression, this allows clients to perform multiple back-and-forth optimistic interactions with the system within the same block.
A client can submit an activation, receive the optimistic result, and use that result to construct a follow-up activation -- all before the first activation is finalized through consensus.

Notably, the correctness of speculative execution depends solely on leader behavior.
Unlike speculative protocols such as Speculative Paxos and NOPaxos that rely on network ordering properties,
our followers verify leader execution through deterministic re-execution.
A speculative result may be discarded if the leader is Byzantine or if leader replacement causes a different execution to be committed.
Network delays or message reordering do not cause speculative results to be incorrect -- they affect only finalization timing.

\section{Discussion}

\label{appendix:design_discussions}

This section discusses design choices and tradeoffs beyond \gen's core architecture.

\subsection{System-Level Optimizations}

\paragraph{Pipelining and Continuous Execution}
\label{appendix:pipelining}

The protocol description presents execution sequentially for clarity, but the system can be pipelined in practice.
Within blocks, many activations are independent due to the execution model and can execute in parallel at both leaders and followers.
The key constraint is maintaining determinism and causality through correct timestamp computation and clear block boundaries.
Assuming each shard hosts multiple accounts, followers can interleave and parallelize execution across accounts within the same shard while preserving the required per-account execution order.

\paragraph{Long-Running Activations}

In \gen{}, an activation's execution need not complete within a single block interval.
Unlike traditional blockchains that use commit-then-execute, where consensus determines the order of transactions before execution begins,
\gen's execute-before-agree pattern decouples execution duration from block timing.

The fundamental challenge in commit-then-execute systems is that consensus commits to a specific set of transactions for each block without knowing their execution time.
Once block $k$ is committed, all transactions in that block must execute to completion before the system can proceed to block $k+1$.
If any transaction takes too long to execute, the system stalls waiting for it to complete.
This forces a bound on per-transaction execution, typically enforced through gas limits or execution time caps.

\gen{} relaxes this constraint because execution happens before consensus, and blocks are determined ``retroactively'' by the leader clock.
The leader executes activations continuously and periodically packages completed work into blocks for consensus.
If an activation does not complete before a block boundary, it is included in a later block, once its execution completes.
The activation is tagged with that block identifier, ensuring any continuations it spawns inherit the correct minimum block id for causal ordering.

Per-activation computation remains bounded -- activations are gas-metered against a declared maximum -- but that bound need not fit within a block interval.

\subsection{Policy and Mechanism Choices}

\paragraph{Fairness Mechanisms}

Blockchains often aim to provide fairness to users through mechanisms such as rotating leaders or restrictions intended to mitigate MEV attacks
such as front-running and sandwich attacks~\cite{daian2019flash,qin2021attacking}.
Fairness is not the focus of this paper.
However, removing atomicity at the infrastructure level permits additional interleavings across activation trees.
A malicious leader may therefore reorder continuations from concurrent activation trees, potentially violating application-level ordering expectations.

Applications requiring ordering guarantees beyond \gen{}'s default per-account execution order can introduce application-specific coordination.
For example, an application can designate a specific account as a \emph{synchronizer} through which relevant continuations are ordered before further execution.
Such synchronization reduces parallelism for the affected execution paths, but allows applications to selectively obtain stronger ordering guarantees when needed.

\paragraph{Follower Verification Strategies}
The follower verification protocol divides verification into two phases: executing activations in the proposal, then verifying continuations.
A coordination mechanism is needed to allow shards to transition from phase one to phase two, ensuring the first phase has completed across all shards in the cluster.
Without this coordination, a shard might begin continuation verification before other shards have sent their continuations, leading to false rejection of valid blocks.

Our implementation uses standard timeout techniques, but alternative approaches exist.
For example, the cluster representative could collect completion signals from shards and coordinate the phase transition.
This approach can identify Byzantine proposals in some scenarios but requires additional synchronization within the cluster during phase two of follower execution.